\documentclass{llncs}

\usepackage{amssymb,mathtools,mathrsfs,enumitem} 
\usepackage{overpic,subcaption}
\usepackage{hyperref}
\usepackage[linesnumbered,ruled,vlined]{algorithm2e}
\usepackage{todonotes}\reversemarginpar
\usepackage{geometry}

\newcommand{\head} {\text{head}}
\newcommand{\tail} {\text{tail}}
\newcommand{\rev} {\text{rev}}

\newcommand{\TOPOLINO} {\texttt{NCSPsupergraph}\xspace}
\newcommand{\MINNIE} {\texttt{NCSPunion}\xspace}

\newcommand{\w}{w}
\newcommand{\esterna}{f^{\infty}}
\newcommand{\uuu}{\overline}

\newcommand{\dd}{\overrightarrow}

\begin{document}
\pagestyle{plain}

\title{Non-Crossing Shortest Paths in Undirected Unweighted Planar Graphs in Linear Time}

\author{Lorenzo Balzotti\inst{1}\orcidID{0000-0001-6191-9801} \and
Paolo G. Franciosa\inst{2}\orcidID{0000-0002-5464-4069} 
\authorrunning{L. Balzotti and P. G. Franciosa}
\institute{Dipartimento di Scienze di Base e Applicate per l’Ingegneria, Sapienza Universit\`a di Roma, Via Antonio Scarpa, 16, 00161 Roma, Italy 
\email{lorenzo.balzotti@sbai.uniroma1.it}
\and
Dipartimento di Scienze Statistiche, Sapienza Universit\`a di Roma, p.le Aldo Moro 5, 00185 Roma, Italy  \email{paolo.franciosa@uniroma1.it}
}
}

%\author{Lorenzo Balzotti\footnote{Dipartimento di Scienze di Base e Applicate per l’Ingegneria, Sapienza Universit\`a di Roma, Via Antonio Scarpa, 16, 00161 Roma, Italy. \texttt{lorenzo.balzotti@sbai.uniroma1.it}.}$\text{ }^{[0000-0001-6191-9801]}$
%\and
%{Paolo G. Franciosa\footnote{Dipartimento di Scienze Statistiche, Sapienza Universit\`a di Roma, p.le Aldo Moro 5, 00185 Roma, Italy. \texttt{paolo.franciosa@uniroma1.it}}}$\text{ }^{[0000-0002-5464-4069]}$}

\date{}
\maketitle
%\todo{tenere traccia di tutte le modifiche, per ora metto un todo bianco}
%\todob{capire Observation 1 del reviewer 3 NON SI CAPISCE...}
\begin{abstract} 
Given a set of well-formed terminal pairs on the external face of an undirected planar graph with unit edge weights, we give a linear-time algorithm for computing the union of non-crossing shortest paths joining each terminal pair, where well-formed means that such a set of non-crossing paths exists. This allows us to compute distances between each terminal pair, within the same time bound.

We also give a novel concept of \emph{incremental shortest path} subgraph of a planar graph, i.e., a partition of the planar embedding in subregions that preserve distances, that can be of interest itself. 
\end{abstract}

\noindent
\textbf{Keywords:} planar graphs, non-crossing paths, shortest paths, undirected unweighted graphs, multiple pairs, external face

\section{Introduction}
\label{section:introduction}
The problem of computing shortest paths in planar graphs arises in application fields such as intelligent transportation system (ITS) and geographic information system (GIS)~\cite{jing-huang,ziliaskopoulos}, route planning~\cite{bauer-delling,goldberg,raney-nagel}, logistic~\cite{masucci-stanilov}, traffic simulations~\cite{baker-gokhale} and robotics~\cite{kim-maxemchuk}. In particular, non-crossing paths in a planar graph are studied to optimize VLSI layout~\cite{bhatt-leighton}, where two \emph{non-crossing} paths may share edges and vertices, but they do not cross each other in the plane.

We are given a planar graph $G=(V,E)$, where $V$ is a set of $n$ vertices and $E$ is a set of edges, with $|E|=O(n)$. The graph has a fixed embedding, and we are also given a set of $k$ terminal pairs $(s_{1},t_{1}),  (s_{2},t_{2}), \ldots, (s_{k},t_{k})$ lying on the external face of $G$.
The non-crossing shortest paths problem (NCSP problem) consists in computing the union of $k$ non-crossing shortest paths in $G$, each joining a terminal pair $(s_{i},t_{i})$, provided that such non-crossing paths exist.

\paragraph{State of the art}
Takahashi \emph{et al.}~\cite{giappo2} solved the NCSP problem in a non-negative edge-weighted planar graph in $O(n\log k)$ worst-case time (actually, in their paper the worst-case time is $O(n\log n)$, that can easily reduced to $O(n\log k)$ by applying the planar single source shortest path algorithm by Henzinger \emph{et al.}~\cite{henzinger}). Their result is improved by Steiger in $O(n\log\log k)$ worst-case time~\cite{steiger}, exploiting the algorithm by Italiano \emph{et al.}~\cite{italiano}. These two algorithms maintain the same time complexity also in the unweighted case.

\paragraph{Our results} In this paper, we solve the NCSP problem on unweighted planar graphs in $O(n)$ worst-case time. We improve, in the unweighted case, the results in~\cite{steiger,giappo2}. By applying the technique in~\cite{err_giappo} we can compute distances between all terminal pairs in linear time.

Our algorithm relies on two main results:

\begin{itemize}
\item an algorithm due to Eisenstat and Klein~\cite{klein1}, that gives in $O(n)$ worst-case time an implicit representation of a sequence of shortest-path trees in an undirected unweighted planar graph $G$, where each tree is rooted in a vertex of the external face of $G$.
Note that, if we want to compute shortest paths from the implicit representation of shortest path trees given in~\cite{klein1},  then we spend $\Theta(kn)$ worst-case time; this happens when all $k$ shortest paths share a subpath of $\Theta(n)$ edges.

\item the novel concept of \emph{incremental shortest paths (ISP) subgraph}  of a graph $G$, introduced in Section~\ref{sec:ISP}, that is a subgraph incrementally built by adding a sequence of shortest paths in $G$ starting from the infinite face of $G$. We show that an ISP subgraph of $G$ partitions the embedding of $G$ into \emph{distance preserving} regions, i.e., for any two vertices $a,b$ in $G$ lying in the same region $R$ it is always possible to find a shortest path in $G$ joining $a$ and $b$ that is contained in $R$.
\end{itemize}

\paragraph{Related work}Our article fits into a wider context of computing many distances in planar graphs. In the positive weighted case, the all pairs shortest paths (APSP) problem is solved by Frederickson in $O(n^2)$ worst-case time~\cite{frederickson}, while the single source shortest paths (SSSP) problem is solved in linear time by Henzinger \emph{et al.}~\cite{henzinger}.
The best known algorithm for computing many distances in planar graphs is due to Gawrychowski \emph{et al.}~\cite{gawrychowski-mozes} and it allows us to compute the distance between any two vertices in $O(\log n)$ worst-case time after a preprocessing requiring $O(n^{3/2})$ worst-case time. In the planar unweighted case, SSSP trees rooted at vertices in the external face can be computed in linear time as in~\cite{klein1}. More results on many distances problem can be found in~\cite{cabello,chen-xu,djidjev,fakcharoenphol-rao,mozes-sommer,nussbaum}.

If we are interested in distances from any vertex in the external face to any other vertex, then we can use Klein's algorithm~\cite{klein2005} that, with a preprocessing of $O(n\log n)$ worst-case time, answers to each distance query in $O(\log n)$ worst-case time.

Kowalik and Kurowski~\cite{kowalik-kurowski} deal with the problem of deciding whether any two query vertices of an unweighted planar graph are closer than a fixed constant $k$. After a preprocessing of $O(n)$ worst-case time, their algorithm answers in $O(1)$ worst-case time, and, if so, a shortest path between them is returned.

Non-crossing shortest paths are also used to compute max-flow in undirected planar graphs~\cite{hassin,hassin-johnson,reif}. In particular, they are used to compute the vitality of edges and vertices with respect to the max-flow~\cite{ausiello-franciosa_2,ausiello-franciosa_1}.

Balzotti and Franciosa~\cite{err_giappo} show that, given the union of a set of shortest non-crossing paths in a planar graph, the lengths of each shortest path can be computed in linear time. This improves the result of~\cite{giappo2}, that can only be applied when the union of the shortest paths is a forest. 

Wagner and Weihe~\cite{wagner-weihe} present an $O(n)$ worst-case time algorithm for finding edge-disjoint (not necessarily shortest) paths in a undirected planar graph such that each path connects two specified vertices on the infinite face of the graph.

\paragraph{Improved results}We specialize the problem of finding $k$ shortest non-crossing path in~\cite{giappo2} to the unweighted case, decreasing the worst-case time complexity from $O(n\log k)$ to $O(n)$ (for every $k$). Therefore, in unweighted graphs we improve the results in~\cite{erickson-nayyeri,kusakari-masubuchi,giappo_rectilinear}.

Erickson and Nayyeri~\cite{erickson-nayyeri} generalized the work in~\cite{giappo2} to the case in which the $k$ terminal pairs lie on $h$ face boundaries. They prove that $k$ non-crossing paths, if they exists, can be found in $2^{O(h^2)}n\log k$ time. Applying our results, if the graph is unweighted, then the time complexity decreases to $2^{O(h^2)}n$ in the worst case.

The same authors of~\cite{giappo2} used their algorithm to compute $k$ non-crossing rectilinear paths with minimum total length in a plane with $r$ obstacles~\cite{giappo_rectilinear}. They found such paths in $O(n\log n)$ worst-case time, where $n=r+k$, which reduces to $O(n)$ worst-case time if the graph is unweighted by using our results.

Kusakari \emph{et al.}~\cite{kusakari-masubuchi} showed that a set of non-crossing forests in a planar graph can be found in $O(n\log n)$ worst-case time, where two forest $F_1$ and $F_2$ are \emph{non-crossing} if for any pair of paths $p_1\subseteq F_1$ and $p_2\subseteq F_2$, $p_1$ and $p_2$ are non-crossing. With our results, if the graph is unweighted, then the time complexity becomes linear.

\paragraph{Our approach} We represent the structure of terminal pairs by a partial order called \emph{genealogy tree}. We introduce a new class of graphs, ISP subgraphs, that partition a planar graph into regions that preserve distances. Our algorithm is split in two parts.

In the first part we use Eisenstat and Klein's algorithm  that gives a sequence of shortest path trees rooted in the vertices of the external face. We choose some specific shortest paths from each tree to obtain a sequence of ISP subgraphs $X_1,\ldots X_k$. By using the distance preserving property of regions generated by ISP subgraphs', we prove that $X_i$ contains a shortest $s_i$-$t_i$ path, for all $i\in\{1,\ldots,k\}$.

In the second part of our algorithm, we extract from each $X_i$ a shortest $s_i$-$t_i$ path and we obtain a set of shortest non-crossing paths that is our goal. In this part we strongly use the partial order given by the genealogy tree.

\paragraph{Structure of the paper} After giving some definitions in Section~\ref{sec:definitions}, in Section~\ref{sec:ISP} we explain the main theoretical novelty.
In Section \ref{sec:our_algorithm} first we resume Eisenstat and Klein's algorithm in Subsection~\ref{section:klein's_algorithm}, then in Subsections~\ref{sec:TOPOLINO} and \ref{sec:computational_complexity} we show the two parts of our algorithm, and we prove the whole computational complexity. Conclusions are given  in Section~\ref{sec:conclusions}.

\section{Definitions}\label{sec:definitions}

Let $G$ be a plane graph, i.e., a planar graph with a fixed planar embedding. We denote by $f^\infty_{G}$ (or simply $f^{\infty}$) its unique infinite face, it will be also referred to as the \emph{external} face of $G$. Given a face $f$ of $G$ we denote by $\partial f$ its boundary cycle. Topological and combinatorial definitions of planar graph, embedding and face can be found in~\cite{gross-tucker}.
 
We recall standard union and intersection operators on graphs, for convenience we define the empty graph as a graph without edges.

\begin{definition}
Given two undirected (or directed) graphs $G=(V(G),E(G))$ and $H=(V(H),E(H))$, we define the following operations and relations:
\begin{itemize}
\item $G\cup H=(V(G)\cup V(H),E(G)\cup E(H))$,
\item $G\cap H=(V(G)\cap V(H),E(G)\cap E(H))$,
\item $G\subseteq H\Longleftrightarrow V(G)\subseteq V(H)$ and $E(G)\subseteq E(H)$,
\item $G\setminus H=(V(G),E(G)\setminus E(H))$,
\item $G=\emptyset\Longleftrightarrow$ $E(G)=\emptyset$ ($V(G)$ can be nonempty).
\end{itemize}
\end{definition}

\noindent
Given an undirected (resp., directed) graph $G=(V(G),E(G))$, given an edge (resp., dart) $e$ and a vertex $v$ we write, for short, $e\in G$ in place of $e\in E(G)$ and $v\in G$ in place of $v\in V(G)$.

We denote by $uv$ the edge whose endpoints are $u$ and $v$ and we denote by $\dd{uv}$ the dart from $u$ to $v$. For every dart $\dd{uv}$ we define $\rev[\dd{uv}]=\dd{vu}$ and $\head[\dd{uv}]=v$.
For every vertex $v\in V(G)$ we define the \emph{degree of $v$} as $deg(v)=|\{e\in E(G) \,|\, \text{$v$ is an endpoint of $e$}\}|$.

For each $\ell\in\mathbb{N}$ we denote by $[\ell]$ the set $\{1,\ldots,\ell\}$.

Given a (possibly not simple) cycle $C$, we define the \emph{region bounded by $C$}, denoted by $R_{C}$, as the maximal subgraph of $G$ whose external face has $C$ as boundary.

\subsection{Paths and non-crossing paths}
\label{section:paths_and_non-crossing_paths}

Given a directed path $p$ we denote by $\uuu{p}$ its undirected version, in which each dart $\dd{ab}$ is replaced by edge $ab$; moreover, we denote by $\rev[p]$ its reverse version, in which each dart $\dd{ab}$ is replaced by dart $\dd{ba}$.

We say that a path $p$ is an \emph{$a$-$b$ path} if its extremal vertices are $a$ and $b$; clearly, if $p$ is a directed path, then $p$ starts in $a$ and it ends in $b$. Moreover, given $i\in[k]$, we denote by \emph{$i$-path} an $s_i$-$t_i$ path, where $s_{i}, t_{i}$ is one of the terminal pairs on the external face.

Given an $a$-$b$ path $p$ and a $b$-$c$ path $q$, we define $p\circ q$ as the (possibly not simple) $a$-$c$ path obtained by the union of $p$ and $q$.

Let $p$ be a simple path and let $a,b\in V(p)$. We denote by $p[a,b]$ the subpath of $p$ with extremal vertices $a$ and $b$.

We denote by $\w(p)$ the length of a path $p$ of a general positive weighted graph $G$. If $G$ is unweighted, then we denote the length of $p$ as $|p|$, that is the number of edges.

We say that two paths in a plane graph $G$ are \emph{non-crossing} if the (undirected) curves they describe in the graph embedding do not cross each other, non-crossing paths may share vertices and/or edges or darts. This property obviously depends on the embedding of the graph;  a combinatorial definition of non-crossing paths can be based on the \emph{Heffter-Edmonds-Ringel rotation principle}~\cite{gross-tucker}.  Crossing and non-crossing paths are given in Figure~\ref{fig:non_crossing_and_single-touch}.

\begin{figure}[h]
\captionsetup[subfigure]{justification=centering}
\centering
%FIGURA 1
	\begin{subfigure}{2.6cm}
\begin{overpic}[width=2.6cm,percent]{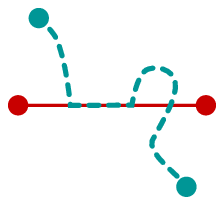}
\end{overpic}
\caption{}\label{fig:non-crossing_a}
\end{subfigure}
\quad
%FIGURA 2
	\begin{subfigure}{2.6cm}
\begin{overpic}[width=2.6cm,percent]{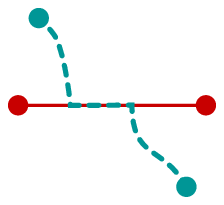}
\end{overpic}
\caption{}\label{fig:non-crossing_b}
\end{subfigure}
\quad
%FIGURA 3
	\begin{subfigure}{2.6cm}
\begin{overpic}[width=2.6cm,percent]{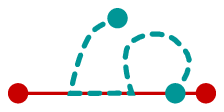}
\end{overpic}
\caption{}\label{fig:non-crossing_c}
\end{subfigure}	
\quad
%FIGURA 4
	\begin{subfigure}{2.6cm}
\begin{overpic}[width=2.6cm,percent]{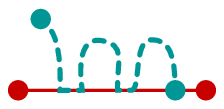}
\end{overpic}
\caption{}\label{fig:non-crossing_d}
\end{subfigure}	
\quad
  \caption{paths in (\subref{fig:non-crossing_a}) and (\subref{fig:non-crossing_b}) are crossing, while paths in (\subref{fig:non-crossing_c}) and (\subref{fig:non-crossing_d}) are non-crossing.} 
\label{fig:non_crossing_and_single-touch}
\end{figure}

\subsection{Genealogy tree}
\label{section:genealogy_tree}

W.l.o.g., we assume that  terminal pairs are distinct, i.e., there is no pair $i,j\in[k]$ such that $\{s_i,t_i\}=\{s_j,t_j\}$. Let $\gamma_i$ be the path in $f^\infty_G$ that goes clockwise from $s_i$ to $t_i$, for $i\in[k]$. We also assume that pairs $\{(s_i,t_i)\}_{i\in[k]}$ are \emph{well-formed}, i.e., for all $j,\ell\in[k]$ either ${\gamma_j}\subset{\gamma_\ell}$ or ${\gamma_j}\supset{\gamma_\ell}$ or ${\gamma_j}\cap{\gamma_\ell}=\emptyset$; otherwise it can be easily seen that it is not possible to find a set of $k$ non-crossing paths joining terminal pairs. This property can be easily verified in linear time, since it corresponds to checking that a string of parentheses is balanced, and it can be done by a sequential scan of the string. 

We define here a partial ordering as in~\cite{err_giappo,giappo2} that represents the inclusion relation between $\gamma_i$'s. This relation intuitively corresponds to an \emph{adjacency} relation between non-crossing shortest paths joining each pair.
Choose an arbitrary $i^*$ such that there are neither $s_{j}$ nor $t_{j}$, with $j\neq i^*$, walking on $\esterna$ from $s_{i^*}$ to $t_{i^*}$ (either clockwise or counterclockwise), and let $e^*$ be an arbitrary edge on that walk. For each $j\in[k]$, we can assume that $e^*\not\in\gamma_j$, indeed if it is not true, then it suffices to switch $s_j$ with $t_j$. We say that  $i \prec j$ if $\gamma_i\subset\gamma_j$. We define the \emph{genealogy tree}  $T_G$ of a set of well-formed terminal pairs as the transitive reduction of poset $([k],\prec)$. W.l.o.g., we assume that $i^*=1$, hence the root of $T_G$ is 1.

If $i\prec j$, then we say that $i$ is a \emph{descendant} of $j$ and $j$ is an \emph{ancestor} of $i$. Moreover, we say that $j$ is the \emph{parent of $i$}, and we write $p(i)=j$, if $i\prec j$ and there is no $r$ such that $i\prec r$ and $r \prec j$.
Figure~\ref{fig:genealogy_tree} shows a set of well-formed terminal pairs, and the corresponding genealogy tree for $i^*=1$.

From now on, in all figures we draw $f^\infty_G$ by a solid light grey line.
W.l.o.g.,  we assume that the external face is a simple cycle, hence, $G$ is a biconnected graph. Indeed, if not, it suffices to solve the NCSP problem in each biconnected component.

\begin{figure}[h]
\captionsetup[subfigure]{justification=centering}
\centering
%FIGURA 1
	\begin{subfigure}{4.5cm}
\begin{overpic}[width=4.5cm,percent]{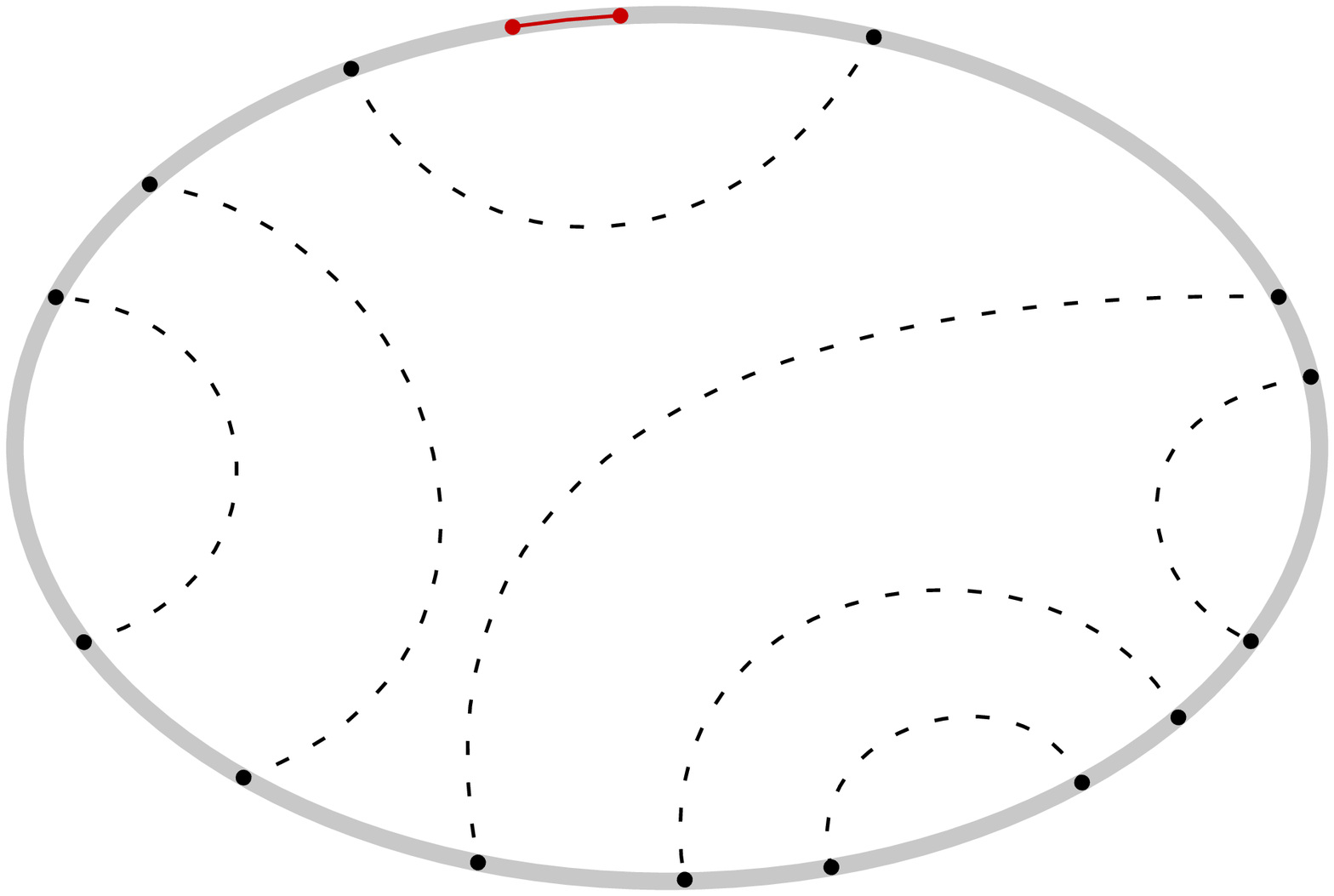}
\put(65,66){$s_1$}
\put(23,65){$t_1$}
\put(98,45){$s_2$}
\put(34.5,-4.5){$t_2$}
\put(100,37.2){$s_3$}
\put(94,12.5){$t_3$}
\put(87.5,7.5){$s_4$}
\put(50,-6){$t_4$}
\put(78.6,2.5){$s_5$}
\put(61,-5.5){$t_5$}
\put(15,3){$s_6$}
\put(4,55){$t_6$}
\put(-2,16.5){$s_7$}
\put(-3,45){$t_7$}

\put(41,67.5){$e^*$}
\end{overpic}
\end{subfigure}
\qquad\qquad
%FIGURA 2
	\begin{subfigure}{1.9cm}
\begin{overpic}[width=1.9cm,percent]{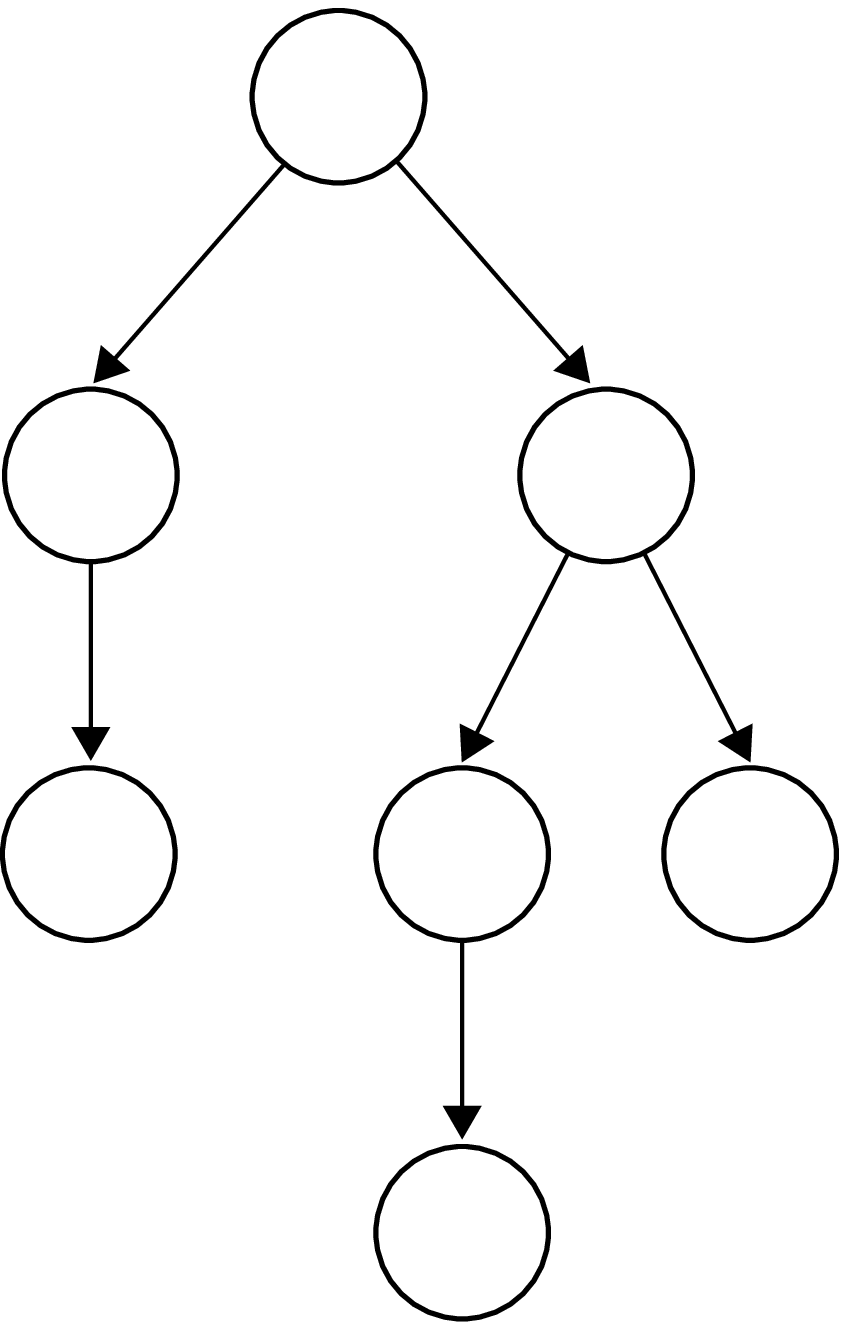}
\put(23.15,89.25){1}
\put(43.5,60.8){2}
\put(4.2,60.8){6}
\put(54.5,31.9){3}
\put(32,31.9){4}
\put(4,31.9){7}
\put(32.5,3){5}
\end{overpic}
\end{subfigure}
\par\smallskip
  \caption{on the left a set of well-formed terminal pairs. Any value in $\{1,3,5,7\}$ can be chosen as $i^*$.  If we choose $i^*=1$, then we obtain the genealogy tree on the right.}
\label{fig:genealogy_tree}
\end{figure}

\section{ISP subgraphs}\label{sec:ISP}

In this section we introduce the concept of \emph{incremental shortest paths (ISP) subgraph} of a graph $G$, that is a subgraph incrementally built by adding a sequence of shortest paths in $G$ starting from $f^\infty_G$ (see Definition~\ref{def:ISP}). The interest towards ISP subgraphs is due to the fact that for any two vertices $a,b$ in $G$ lying in a same face $f$ of the ISP subgraph there is always a shortest path in $G$ joining $a$ and $b$ contained in $f$ (boundary included). All the results of this section hold for positive edge weighted graphs, where the length of a path is the sum of edge weights instead of the number of edges.

This is the main novel result of this paper, that allows us to prove that, in order to build the union of shortest paths joining terminal pairs, we can start from the union of some of the shortest paths computed by the algorithm in~\cite{klein1}.

\begin{definition}\label{def:ISP}
A graph $X$  is an \emph{incremental shortest paths (ISP) subgraph} of a positive weighted graph $G$ if $X=X_r$, where
%$\{X_i\}_{i\in[r]}$ 
$X_{0}$, $X_{1}$, \ldots, $X_{r}$ is a sequence of subgraphs of $G$ built in the following way: $X_0=f^\infty_G$ and $X_i=X_{i-1}\cup p_i$, where $p_i$ is a shortest $x_i$-$y_i$ path in $G$ with  $x_i,y_i\in X_{i-1}$. 
\end{definition}

\begin{remark}\label{remark:degree_1}
All degree one vertices of an ISP subgraph of $G$ are in $f^\infty_G$.
\end{remark}

We define now operator $\downarrow$, that given a path $\pi$ and a cycle $C$, in case $\pi$ crosses $C$, replaces some subpaths of $\pi$ by some portions of $C$, as depicted in Figure~\ref{fig:example_ISP+downarrow}(\subref{fig:downarrow}). We observe that $\pi\downarrow \partial f$ could be not a simple path even if $\pi$ is.

\begin{definition}\label{def:downarrow}
Let $C$ be a cycle in $G$. Let $a,b$ be two vertices in $R_{C}$ and let $\pi$ be a simple $a$-$b$ path. In case $\pi\subseteq R_C$ we define  $\pi\downarrow C=\pi$. Otherwise, let $(v_1,v_2,\ldots,v_{2r})$ be the ordered subset of vertices of $\pi$ that satisfies the following:  $\pi[a,v_1]\subseteq R_C$, $\pi[v_{2r},b]\subseteq R_C$, $\pi[v_{2i-1},v_{2i}]\cap R_C=\emptyset$ and $\pi[v_{2i},v_{2i-1}]\subseteq R_C$, for all $i\in[r]$. For every $i\in[r]$, let $\mu_i$ be the $v_{2i-1}$-$v_{2i}$ path on $C$ such that the region bounded by $\mu_i\circ\pi[v_{2i-1},v_{2i}]$ does not contain $R_C$. We define $\pi\downarrow C=\pi[a,v_1]\circ\mu_1\circ\pi[v_2,v_3]\circ\mu_2\ldots\circ\pi[v_{2r-2},v_{2r-1}]\circ\mu_{r}\circ\pi[v_{2r},b]$.
\end{definition}

Definition~\ref{def:ISP} and Definition~\ref{def:downarrow} are depicted in 
Figure~\ref{fig:example_ISP+downarrow}.

\begin{figure}[h]
\captionsetup[subfigure]{justification=centering}
\centering
%FIGURA 1
	\begin{subfigure}[t]{4.5cm}
\begin{overpic}[width=4.5cm,percent]{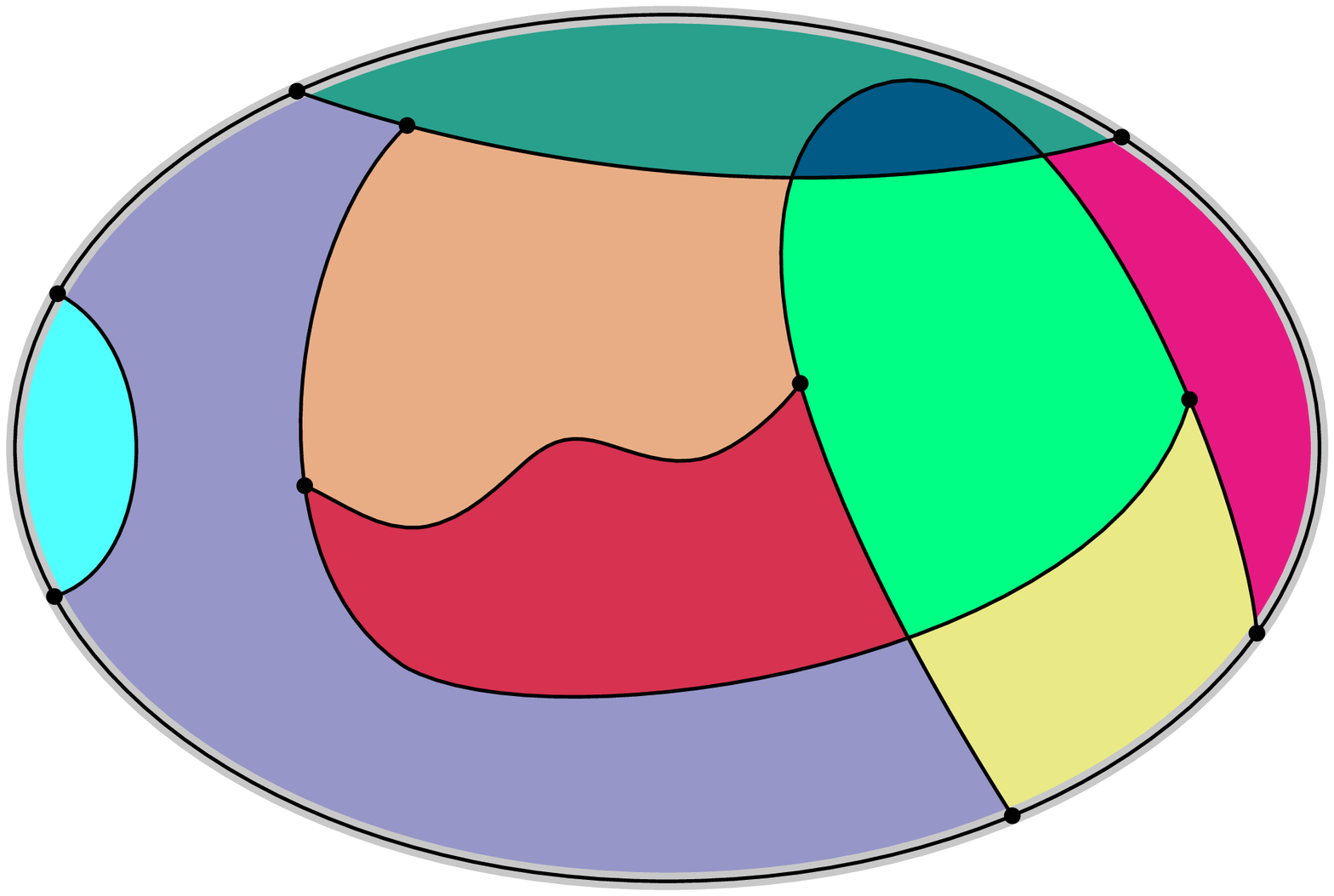}
\put(75,.5){$x_1$}
\put(95,15.2){$y_1$}
\put(-3,17.5){$x_2$}
\put(-4.5,44.4){$y_2$}
\put(31.6,58.5){$x_4$}
\put(90.5,36.5){$y_4$}
\put(85,58){$x_3$}
\put(17.5,63.5){$y_3$}
\put(61,38){$x_5$}
\put(15.2,29.5){$y_5$}
\end{overpic}
\caption{}\label{fig:ISP}
\end{subfigure}
\qquad\qquad
%FIGURA 2
	\begin{subfigure}[t]{5.5cm}
\begin{overpic}[width=5.5cm,percent]{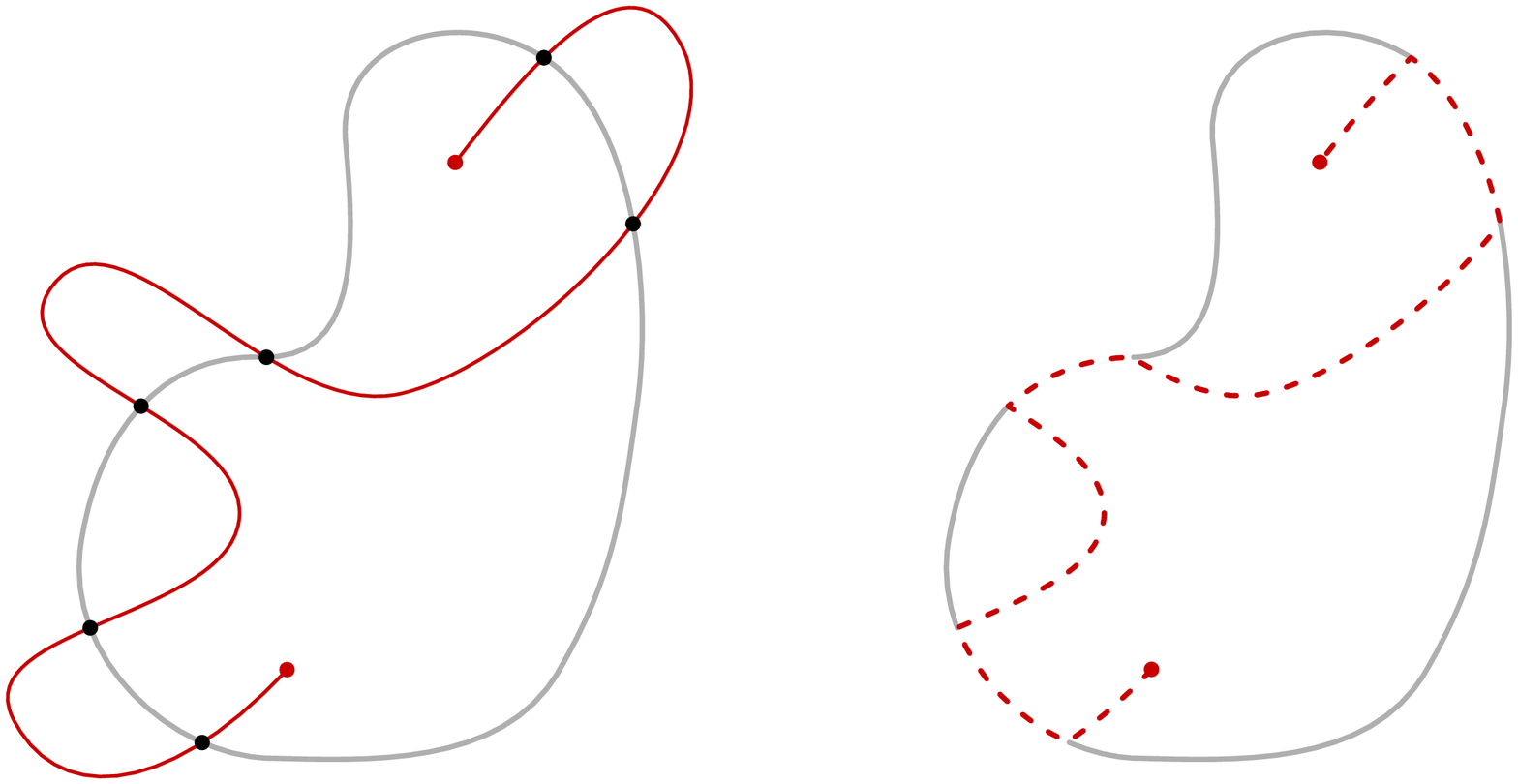}
\put(31,15){$C$}
\put(3.5,35.5){$\pi$}
\put(61.5,31){$\pi\downarrow C$}

\put(19,8){$a$}
\put(26,36){$b$}

\put(12.5,-2.5){$v_1$}
\put(-1.5,11.5){$v_2$}
\put(1,23){$v_3$}
\put(14.5,30.5){$v_4$}
\put(43,32){$v_5$}
\put(30.5,51.2){$v_6$}

\end{overpic}
\caption{}\label{fig:downarrow}
\end{subfigure}	
  \caption{ (\subref{fig:ISP}) an ISP subgraph $X$ of $G$; extremal vertices $x_{i}, y_{i}$ of $p_i$ are drawn, for $i \in [5]$. Different faces of $X$ have different colors. An example of Definition~\ref{def:downarrow} is given in (\subref{fig:downarrow}).}
\label{fig:example_ISP+downarrow}
\end{figure}

In the following theorem we show that, given any face $f$ of an ISP subgraph $X$ of $G$,  every path $\pi$ in $G$ whose extremal vertices are in $R_{\partial f}$ is not shorter than $\pi\downarrow \partial f$.

\begin{theorem}\label{prop:main} 
Let $X$ be an ISP subgraph of $G$. Let $f$ be any face of $X$, and let $a,b$ be two distinct vertices in $R_{\partial f}$. For any $a$-$b$ path $\pi$ we have $\w(\pi\downarrow \partial f) \leq \w(\pi)$.
\end{theorem}
\begin{proof}
Let $\{X_i\}_{i\in[r]}$ be the sequence of ISP subgraphs such that $X=X_r$, and let $p_i$ be the path that builds $X_i$ from $X_{i-1}$. We assume that $p_{i}$ has no vertices in $X_{i-1}$ other than its endpoints $x_i$ and $y_i$, otherwise we can split $p_{i}$ on intersections with $X_{i-1}$ and repeatedly apply the same proof to each portion of $p_{i}$. We prove the thesis by induction on $j$ for every choice of a face $f$ of $X_j$, $a,b\in R_{\partial f}$ and $a$-$b$ path $\pi$.

In the base case, where $j=1$, there are exactly two faces $A$ and $B$ in $X_1$ other than $f^\infty_G$. Let $a,b\in  V(R_{\partial A})$ (the same argument holds for $B$) and let $\pi$ be any $a$-$b$ path. In case $\pi \subseteq R_{\partial A}$ we have $\pi\downarrow \partial A = \pi$, hence the thesis trivially holds. In case  $\pi \not\subseteq R_{\partial A}$, then $\pi\downarrow \partial A$ is not longer than $\pi$ because some subpaths of $\pi$ have been replaced by subpaths of $p_1$ with the same extremal vertices and $p_1$ is a shortest path.

We assume that the thesis holds for all $i<j$ and we prove it for $j$. Let $f$ be a face of $X_j$ and let $f'$ be the unique face of $X_{j-1}$ such that $f\subset f'$ (Figure~\ref{fig:f_f'}(\subref{fig:f_f'_1}) and ~\ref{fig:f_f'}(\subref{fig:f_f'_2}) show faces $f$ and $f'$, respectively). Let $a,b\in V(R_{\partial f})$ and let $\pi$ be an $a$-$b$ path. Three cases may occur:
\begin{itemize}
\item \textbf{case $\pi\subseteq R_{\partial f}$:}
the thesis trivial holds, since $\pi\downarrow \partial f=\pi$;

\item \textbf{case $\pi\subseteq R_{\partial f'}$ and $\pi\not\subseteq R_{\partial f}$:} since $\pi\subseteq R_{\partial f'}$ and $\pi\not\subseteq R_{\partial f}$, then $\pi$ crosses $p_j$ an even number of times, thus $\pi\downarrow \partial f$ is not longer than $\pi$, since some subpaths of $\pi$ have been replaced by subpaths of $p_j$ with the same extremal vertices and $p_j$ is a shortest path (see Figure~\ref{fig:f_f'}(\subref{fig:f_f'_3}) where $\pi$ is the red and dashed path);

\item \textbf{case $\pi\not\subseteq R_{\partial f'}$:} since $f \subseteq f'$, it is easy to see that $\pi \downarrow \partial f = (\pi \downarrow \partial f') \downarrow \partial f$. Let us consider $\pi' = \pi\downarrow \partial f'$. By induction, it holds that $\w(\pi') \leq \w(\pi)$. We observe now that $\pi'\subseteq R_{\partial f'}$ and $\pi'\not\subseteq R_{\partial f}$, hence the previous case applies, showing that $\w(\pi' \downarrow \partial f) \leq \w(\pi')$.
Finally, the two previous inequalities imply $\w(\pi \downarrow \partial f) \leq \w(\pi \downarrow \partial f') \leq \w(\pi)$ (see Figure~\ref{fig:f_f'}(\subref{fig:f_f'_3}) where $\pi$ is the green and continue path).\qed
\end{itemize}
%\vspace{-20pt}
\end{proof}

\begin{figure}[h]
\captionsetup[subfigure]{justification=centering}
\centering
%FIGURA 1
	\begin{subfigure}{4.2cm}
\begin{overpic}[width=4.2cm,percent]{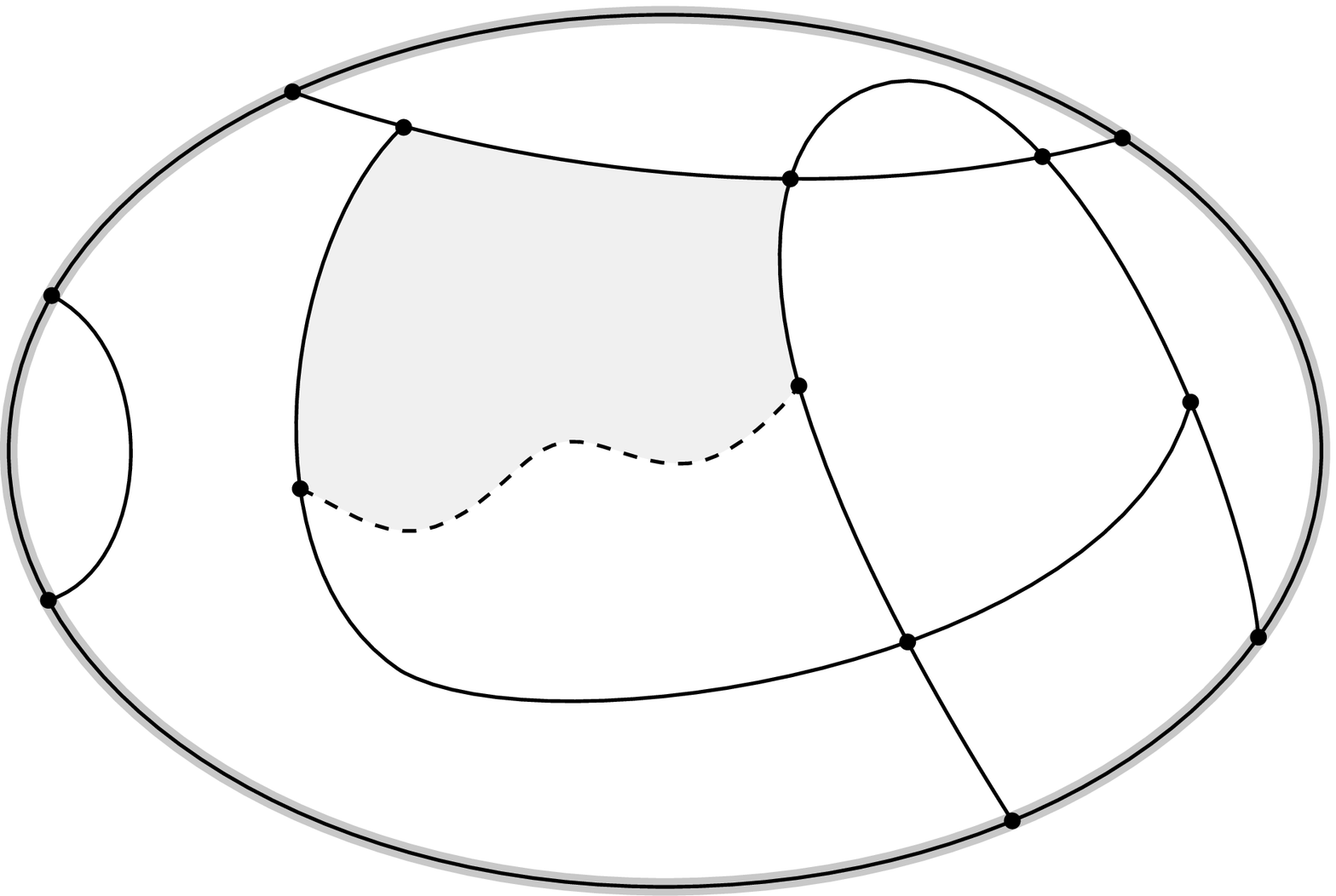}
\put(40,43){$f$}
\put(40,27){$p_j$}
\put(5.5,61){$X_j$}
\end{overpic}
\caption{a face $f$ in $X_j$ \\$\mathbf{}$ \\$\mathbf{}$}\label{fig:f_f'_1}
\end{subfigure}
\quad
%FIGURA 2
	\begin{subfigure}{4.2cm}
\begin{overpic}[width=4.2cm,percent]{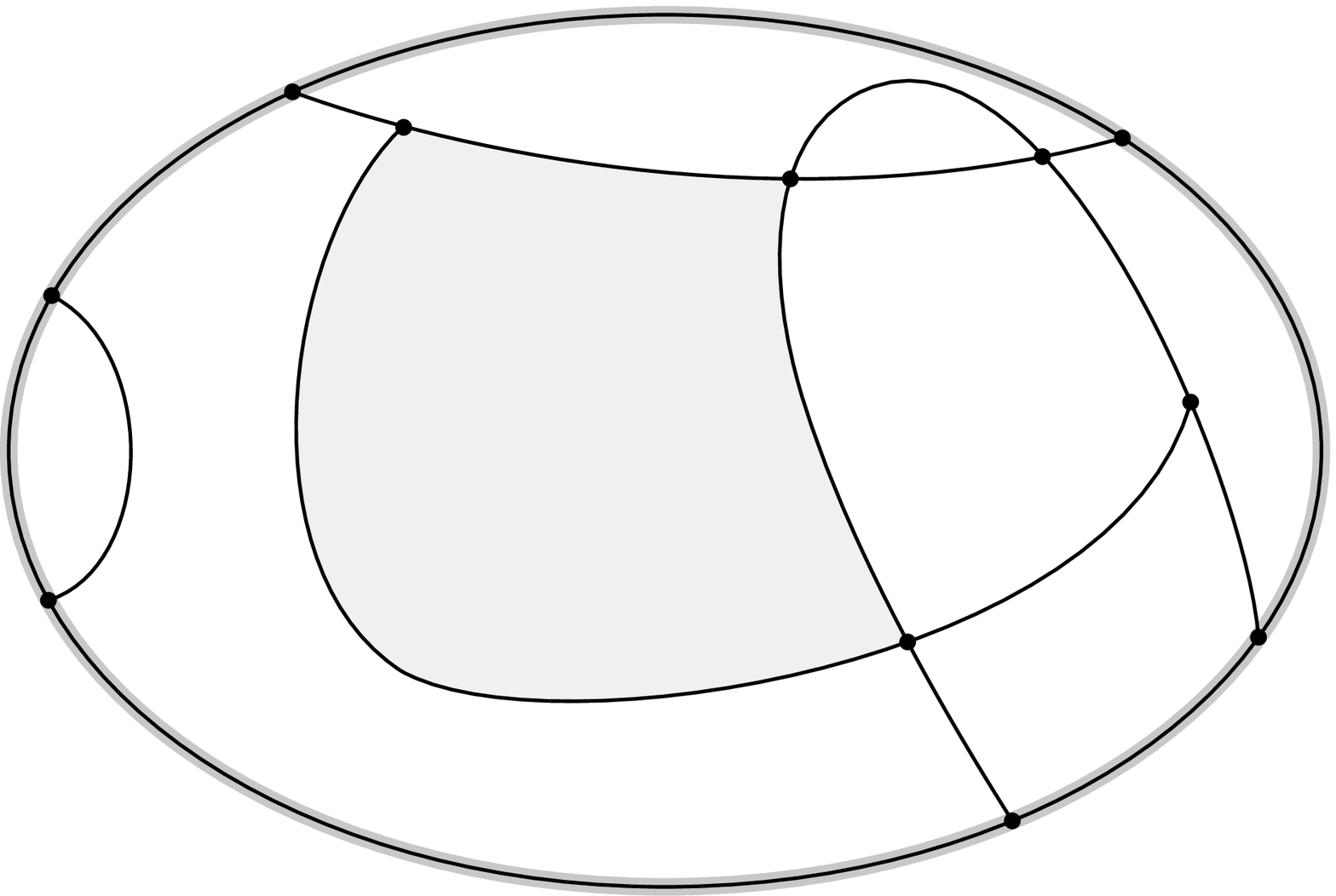}
\put(42,35){$f'$}
\put(0,61){$X_{j-1}$}
\end{overpic}
\caption{a face $f'$ in $X_{j-1}$ \\$\mathbf{}$ \\$\mathbf{}$}\label{fig:f_f'_2}
\end{subfigure}	
\quad
	\begin{subfigure}{4.2cm}
\begin{overpic}[width=4.2cm,percent]{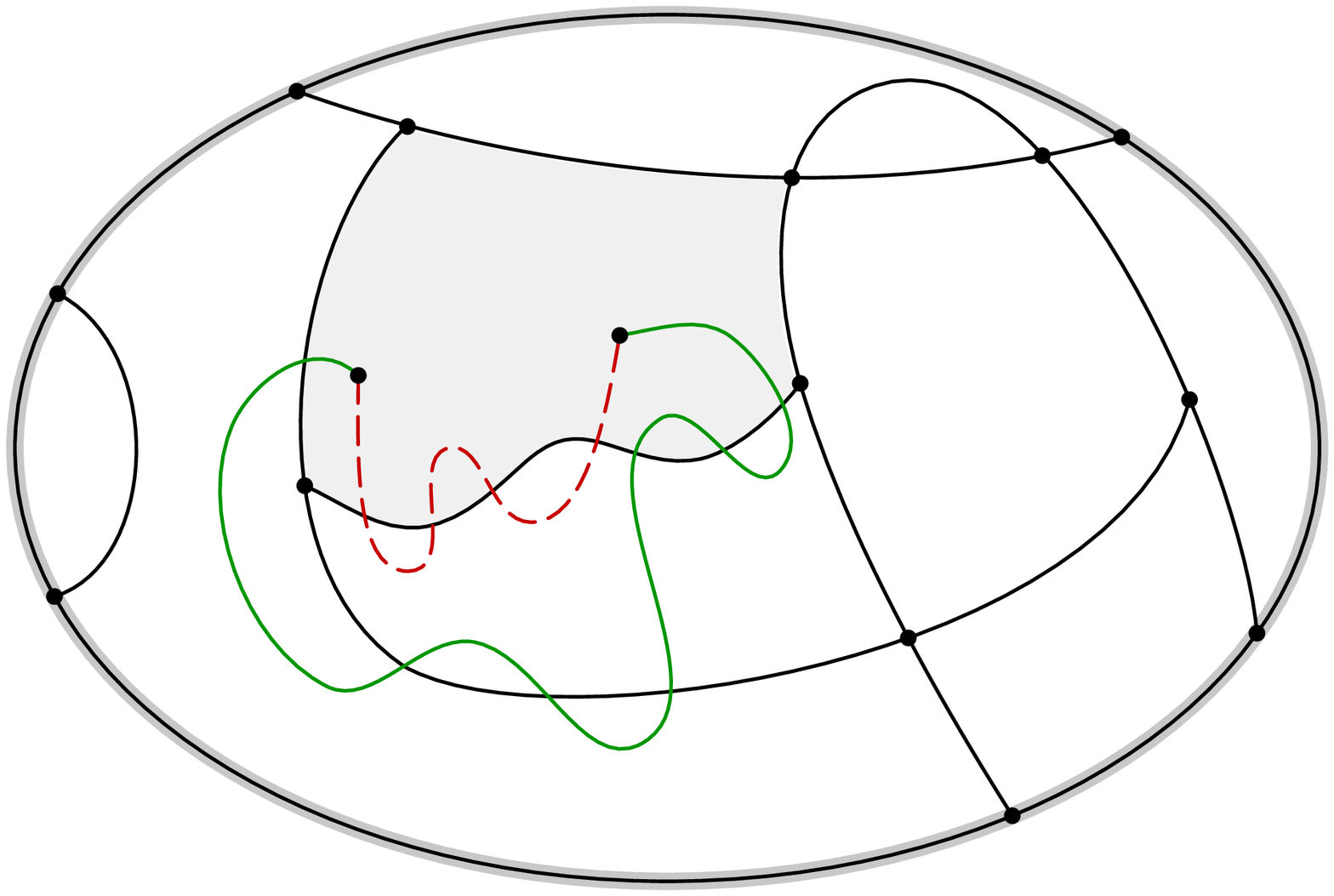}
\put(5.5,61){$X_j$}

\put(27,41){$a$}
\put(41,42){$b$}

\end{overpic}
\caption{two examples of $\pi$ in the second case (dashed red) and third case (continuous green)}\label{fig:f_f'_3}
\end{subfigure}	
\caption{in (\subref{fig:f_f'_1}) and (\subref{fig:f_f'_2}) faces $f$ and $f'$ build on the ISP graph in Figure~\ref{fig:example_ISP+downarrow}(\subref{fig:ISP}). In (\subref{fig:f_f'_3}) we depict the second and third case of the proof of Theorem~\ref{prop:main}.}
\label{fig:f_f'}
\end{figure}

We can state now the main property of ISP subgraphs.
\begin{corollary}
Let $X$ be an ISP subgraph of $G$ and let $f$ be any face of $X$. For every $a,b\in R_{\partial f}$ there exists a shortest $a$-$b$ path of $G$ contained in $R_{\partial f}$.
\end{corollary}

\section{Our algorithm}\label{sec:our_algorithm}

We summarize in Subsection~\ref{section:klein's_algorithm} the result of Eisenstat and Klein's paper~\cite{klein1}, that deals with the multiple-source shortest paths problem. For the sake of clarity, we split our algorithm in two parts:
\begin{itemize}
\item in Subsection~\ref{sec:TOPOLINO} we introduce algorithm \TOPOLINO, that builds a sequence $\{X_i\}_{i\in[k]}$ of subgraphs of $G$ such that $X_k$ contains a shortest path for each terminal pair, and it possibly contains some extra edges. We anticipate that $X_i\cup f^\infty_G$ is an ISP subgraph of $G$, for all $i\in[k]$. 
\item in Subsection~\ref{sec:computational_complexity} we present algorithm \MINNIE that, by using algorithm \TOPOLINO, builds a directed graph that is exactly the union of the shortest directed paths joining each terminal pair contained in the output of algorithm \TOPOLINO.
\end{itemize}

\subsection{Eisenstat and Klein's result}
\label{section:klein's_algorithm}

The algorithm in~\cite{klein1} takes as input an undirected unweighted planar graph $G$, where $v_1, v_{2},\ldots,v_r$ is the sequence of vertices in the external face of $G$  in clockwise order,  and returns an implicit representation of a sequence of shortest path trees $\mathcal{T}_i$,  for $i\in[r]$, where each $\mathcal{T}_i$ is rooted in $v_{i}$.

The sequence of trees $\mathcal{T}_i$, for $i\in[r]$, is represented by explicitly listing the darts in $\mathcal{T}_1$, and listing the darts that are added to transform $\mathcal{T}_i$ into $\mathcal{T}_{i+1}$, for $1 < i \leq r$ (for each added dart from $x$ to $y$, the unique dart that goes to $y$ in $\mathcal{T}_i$ is deleted; with the only two exceptions of the added dart leading to $v_{i}$, and the deleted dart leading to $v_{i+1}$). Hence, the output of their algorithm is $\mathcal{T}_1$ and a sequence of sets of darts. A key result in~\cite{klein1} shows that if a dart $d$ appears in $\mathcal{T}_{i+1}\setminus\mathcal{T}_i$, then $d$ cannot appear in any $\mathcal{T}_{j+1}\setminus\mathcal{T}_j$, for $j>i$. Thus the implicit representation of the sequence of shortest path trees has size $O(n)$. This representation can be computed in  $O(n)$ worst-case time.

\subsection{Algorithm \TOPOLINO}
\label{sec:TOPOLINO}

Algorithm \TOPOLINO  builds a sequence $\{X_i\}_{i\in[k]}$ of subgraphs of $G$ by using the sequence of shortest path trees given by Eisenstat and Klein's algorithm. We point out that we are not interested in the shortest path trees rooted at every vertex of $f^\infty_G$, but we only need the shortest path trees rooted in $s_i$'s. So, we define $T_i$ as the shortest path tree rooted in $s_i$, for $i\in[k]$. We denote by $T_i[v]$ the path in $T_i$ from $s_i$ to $v$.

The algorithm starts by computing the first subgraph $X_1$, that is just the undirected 1-path in $T_1$, i.e., $\uuu{T_{1}[t_{1}]}$ (we recall that all $T_i$'s trees given by algorithm in~\cite{klein1} are rooted directed tree). Then the sequence of subgraphs $X_i$, for $i=2,\ldots,k$ is computed by adding some undirected paths extracted from the shortest path trees $T_i$'s defined by Eisenstat and Klein's algorithm.

We define the set $H_{i}\subseteq X_i$ of vertices $h$ such that at least one dart $d$ is added while passing from $T_{i-1}$ to $T_{i}$ such that $\head[d]=h$. Hence, $H_i$ is the set of vertices whose parent in $T_{i}$ differs from the parent in $T_{i-1}$.
At iteration $i$, we add path $\uuu{T_i[h]}$ to $X_i$, for each $h$ in $H_i$.
%\todo{non e' meglio aggiungere una prima riga su Eisenstat and Klein? NOOOO!!! Infatti eseguo Klein durante l'algoritmo, forse va detto... forse andrebbe scritto nell'algoritmo mettendo la prima riga ``compute $T_1$" e poi ``compute $T_i$"}
%\todo{OSS: se metto il calcolo dei $T_i$ allora potrei scrivere esplicitamente gli $H_i$. MMMMMMMMMMM: $H_i$ is the set of vertices (OF $X_i$!!!) whose parent in $T_{i}$ differs from the parent in $T_{i-1}$}

\begin{figure}[h]
\begin{algorithm}[H] 
\SetAlgorithmName{Algorithm \texttt{NCSPsupergraph}}{}{}
\renewcommand{\thealgocf}{}
 \caption{}
 \KwIn{an undirected unweighted planar embedded graph $G$ and $k$ well-formed terminal pairs of vertices $(s_i,t_i)$, for $i\in[k]$, on the external face of $G$}
 \KwOut{an undirected graph $X_k$ that contains a set of non-crossing paths $P=\{\pi_1,\ldots,\pi_k\}$, where $\pi_i$ is a shortest $s_i$-$t_i$ path, for $i\in[k]$}
{$X_1=\overline {T_1[t_1]}$\label{line:1_compute_pi_1}\;
\For{$i=2,\ldots,k$}{
$X_i=X_{i-1}$\label{line:X_i=X_i-1}\;
For all $h\in H_i$, $X_i=X_i\cup \overline{T_i[h]}$\label{line:d2}\;
Let $\eta_i$ be the undirected path on $T_i$ that starts in $t_i$ and walks backwards until a vertex in $X_i$ is reached\label{line:mu_i}\;
$X_i=X_i\cup \eta_i$\label{line:X+mu_i}\;
}
}
\end{algorithm}
\end{figure}

\begin{lemma}\label{lemma:TOPOLINO_O(n)}
Algorithm \TOPOLINO has $O(n)$ worst-case time complexity.
\end{lemma}
\begin{proof}
Eisenstat and Klein's algorithm requires $O(n)$ worst-case time, implying that the $H_i$'s and the $T_i$'s can be found in $O(n)$ worst-case time. Algorithm \TOPOLINO visits each edge of $G$ at most $O(1)$ times (in Line~\ref{line:d2}, $\overline{T_i[h]}$ can be found by starting in $h$ and by walking backwards on $T_i$ until a vertex of $X_i$ is found). The thesis follows.\qed
\end{proof}

Figure~\ref{fig:TOPO} shows how algorithm \TOPOLINO builds $X_4$ starting from $X_3$. Starting from $X_{3}$ in Figure~\ref{fig:TOPO}(\subref{fig:TOPO_1}), Figure~\ref{fig:TOPO}(\subref{fig:TOPO_2}) shows the darts whose head is in $H_4$. Consider the unique dart $d$ whose head is the vertex $x$: we observe that $\uuu{d}$ is already in $X_3$, this happens because $\rev[d]\in T_3[t_3]$. Indeed, it is possible that at iteration $i$ some portions of some undirected paths that we add in Line~\ref{line:d2} are already in $X_{i-1}$. Figure \ref{fig:TOPO}(\subref{fig:TOPO_3}) highlights 
$\bigcup_{h\in H_4}T_4[h]$ and $\eta_4$, while in Figure \ref{fig:TOPO}(\subref{fig:TOPO_4}) $X_4$ is drawn.

\begin{figure}[h]
\captionsetup[subfigure]{justification=centering}
\centering
%FIGURA 1
	\begin{subfigure}{4.5cm}
\begin{overpic}[width=4.5cm,percent]{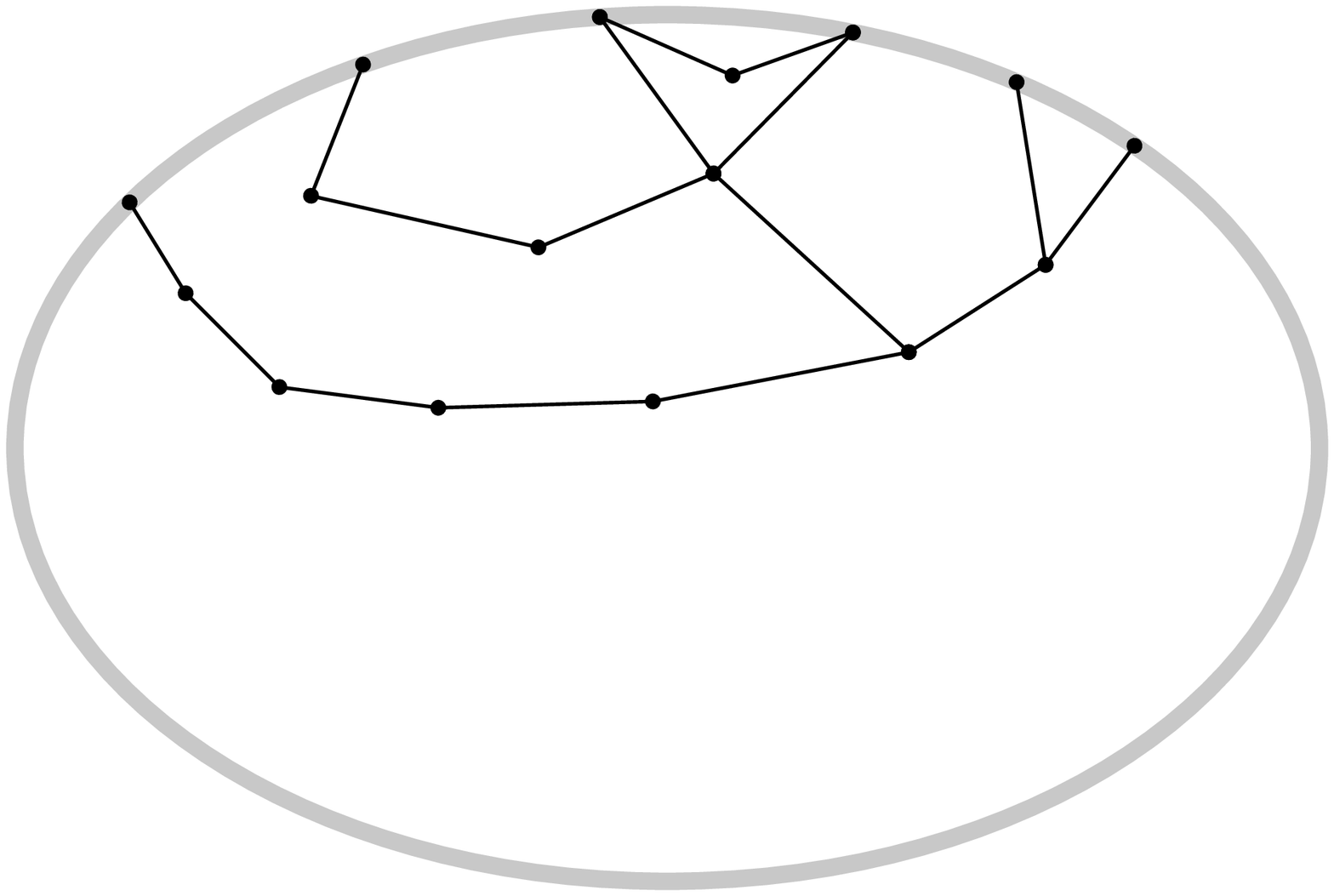}
\put(44,69){$s_1$}
\put(25,66){$t_1$}
\put(62,68){$s_2$}
\put(76,63){$t_2$}
\put(86,58){$s_3$}
\put(5,56){$t_3$}
\end{overpic}
\caption{$X_{3}$ in black\\$\mathbf{}$}\label{fig:TOPO_1}
\end{subfigure}
\qquad
%FIGURA 2
	\begin{subfigure}{4.5cm}
\begin{overpic}[width=4.5cm,percent]{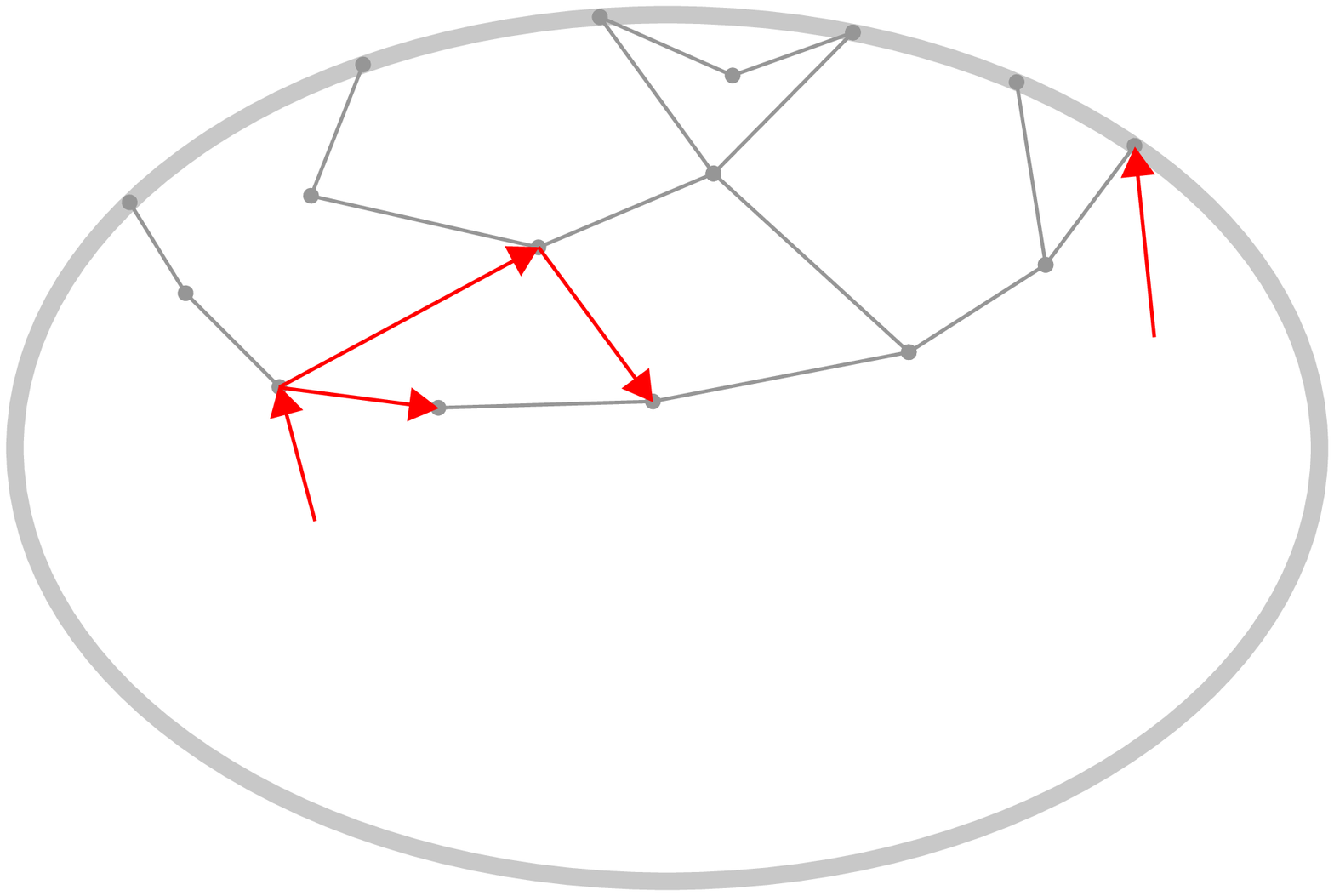}
\put(44,69){$s_1$}
\put(25,66){$t_1$}
\put(62,68){$s_2$}
\put(76,63){$t_2$}
\put(86,58){$s_3$}
\put(5,56){$t_3$}
\put(31,31.5){$x$}
\end{overpic}
\caption{$X_{3}$ in grey and the darts whose head is in $H_4$ in red}\label{fig:TOPO_2}
\end{subfigure}
\qquad
\par\bigskip
%FIGURA 3
	\begin{subfigure}{4.5cm}
\begin{overpic}[width=4.5cm,percent]{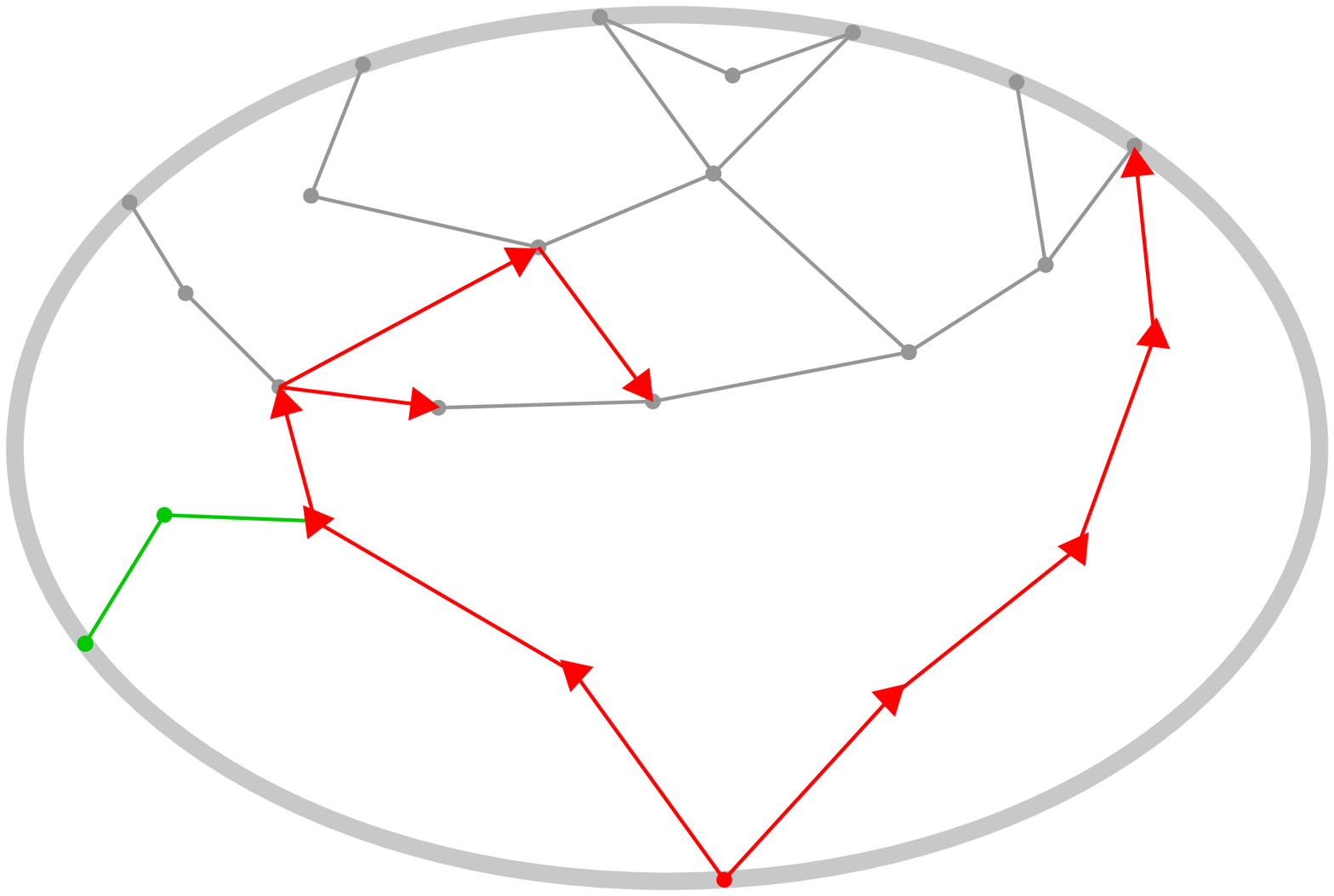}
\put(44,69){$s_1$}
\put(25,66){$t_1$}
\put(62,68){$s_2$}
\put(76,63){$t_2$}
\put(86,58){$s_3$}
\put(5,56){$t_3$}
\put(53,-4.5){$s_4$}
\put(3,12){$t_4$}
\end{overpic}
\par\medskip
\caption{$\bigcup_{h\in H_4}T_4[h]$ in red and $\eta_4$ in green}\label{fig:TOPO_3}
\end{subfigure}	
\qquad
%FIGURA 4
	\begin{subfigure}{4.5cm}
\begin{overpic}[width=4.5cm,percent]{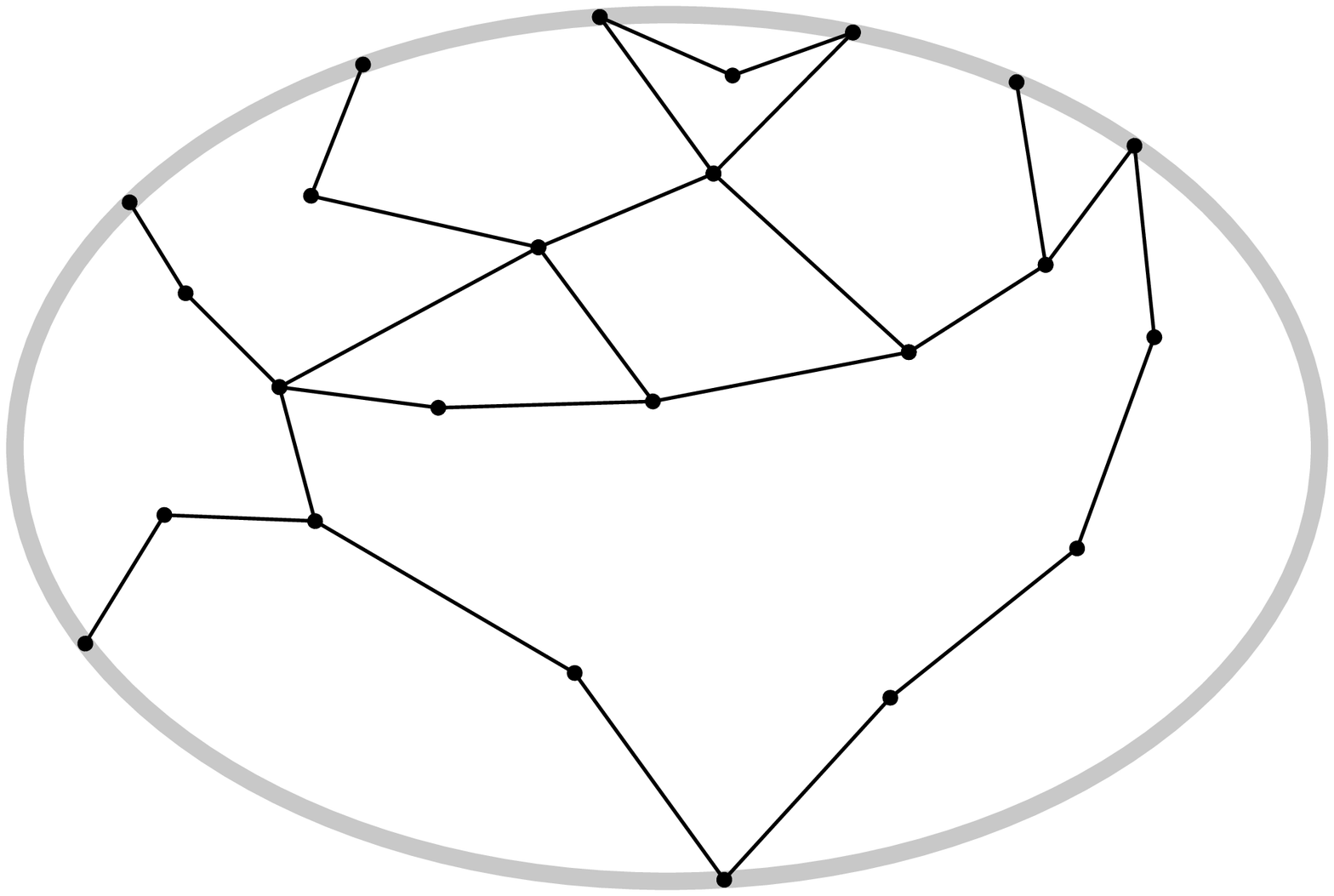}
\put(44,69){$s_1$}
\put(25,66){$t_1$}
\put(62,68){$s_2$}
\put(76,63){$t_2$}
\put(86,58){$s_3$}
\put(5,56){$t_3$}
\put(53,-4.5){$s_4$}
\put(3,12){$t_4$}
\end{overpic}
\par\medskip
\caption{$X_4$ in black\\$\mathbf{}$}\label{fig:TOPO_4}
\end{subfigure}	
  \caption{algorithm \TOPOLINO: graph $X_4$ is built starting from $X_3$.}
\label{fig:TOPO}
\end{figure}

Subgraphs  $\{X_i\}_{i\in[k]}$ built by algorithm \TOPOLINO, together with $f^\infty_G$, satisfy all the hypothesis of Theorem~\ref{prop:main}. Indeed, paths added in Line~\ref{line:d2} and Line~\ref{line:X+mu_i} are shortest paths in $G$ joining vertices in $X_{i-1}$, thus fulfilling Definition~\ref{def:ISP}. So, we exploit Theorem~\ref{prop:main} to prove that $X_i$ contains an $i$-path, for $i\in[k]$, and, in particular, $X_k$ contains a set of non-crossing paths $P=\{\pi_1,\ldots,\pi_k\}$, where $\pi_i$ is a shortest $i$-path, for $i\in[k]$. The main idea is to show that $X_i$ contains an undirected path that has the same length as the shortest $i$-path found by the algorithm by Eisenstat and Klein. This is proved in Theorem~\ref{th:TOPOLINO}.

Given a subgraph $X$ of $G$, we say that an $i$-path $p$ is the \emph{leftmost $i$-path in $X$} if for every $i$-path $q\subseteq X$ it holds $R_{p\circ\gamma_i}\subseteq R_{q\circ\gamma_i}$.

We say that an undirected $a$-$b$ path $p$ \emph{always turns left} if $p$ choose the leftmost edge, w.r.t. the fixed embedding, in each vertex going from $a$ to $b$.

\begin{theorem}\label{th:TOPOLINO}
Let $\pi_i$ be the leftmost $i$-path in $X_i$, for $i\in[k]$. The following hold:
\begin{enumerate}[label=\thetheorem.(\arabic*), ref=\thetheorem.(\arabic*),leftmargin=\widthof{1.(1)}+\labelsep]
\item\label{item:leftmost} $\pi_i$ is the $s_i$-$t_i$ path in $X_i$ that always turns left, for $i\in[k]$,
\item\label{item:shortest} $\pi_i$ is a shortest $i$-path, for $i\in[k]$,
\item\label{item:pi_i_non-croossing} for all $i,j\in[k]$, $\pi_i$ and $\pi_j$ are non-crossing.
\end{enumerate}
\end{theorem}

\begin{proof} 
\begin{itemize}
\item\textbf{\ref{item:leftmost}} For convenience, for every $i\in[k]$, let $\lambda_i$ be the undirected path on $X_i$ that starts in $s_i$ and always turns left until  it reaches either $t_i$ or a vertex $x$ of degree one in $X_i$; we observe that $\lambda_i$ is well defined and, by Remark~\ref{remark:degree_1}, $x\in f^\infty_G$. We have to prove that $\lambda_i=\pi_i$.

Let $i\in[k]$. First, we observe that $s_i\in X_i$ because $s_{i-1}\in H_i$, thus, by Line~\ref{line:d2}, $\overline{T_i[s_{i-1}]}\subseteq X_i$. This implies $s_i\in X_i$ as we have claimed.

Let $x$ be the extremal vertex of $\lambda_i$ other than $s_i$. Assume by contradiction that $x\neq t_i$. Two cases are possible: either $x\in V(f^\infty_G)\setminus V(\gamma_i)$ or $x\in V(\gamma_i)\setminus\{t_i\}$.

The first case cannot occur because Line~\ref{line:d2} and Line~\ref{line:X+mu_i} imply $\overline{T_i[t_i]}\subseteq X_i$, thus $\lambda_i$ would cross $\eta_i$, absurdum. In the second case, let us assume by contradiction that $x\in V(\gamma_i)\setminus\{t_i\}$. Let $d\in\lambda_i$ be the dart such that $\head[d]=x$. By definition of $\lambda_i$, vertex $x$ has degree one in $X_i$. By Line~\ref{line:1_compute_pi_1}, Line~\ref{line:d2} and Line~\ref{line:X+mu_i}, all vertices with degree one are equal to either $s_\ell$ or $t_\ell$, for some $\ell\in[k]$, and this implies that there exists $j<i$ such that $x\in\{s_j,t_j\}$. This is absurdum because there is not $s_j$ or $t_j$ in $V(\gamma_i)\setminus\{s_i,t_i\}$ such that $j<i$. Hence $\lambda_i$ is an $i$-path, and, by its definition, $\lambda_i$ is the leftmost $i$-path in $X_i$. Therefore $\lambda_i=\pi_i$.

\item\textbf{\ref{item:shortest}} We prove that $\pi_i$ is a shortest $i$-path by using Theorem~\ref{prop:main}, indeed, $X_i\cup f^\infty_G$ is an ISP subgraph of $G$ by construction. Let $G'$ be the graph obtained from $G$ by adding a dummy path $q$ from $s_i$ to $t_i$ in $f^\infty_G$ with high length (for example,  $|q|=|E(G)|$). Let $C$ be the cycle $\pi_i\circ q$. We observe that $\overline{T_i[t_i]}\downarrow C=\pi_i$ and $C$ is the boundary of a face of $G'$. Thus, by Theorem~\ref{prop:main}, $|\pi_i|\leq|\overline{T_i[t_i]}|$. Since $\overline{T_i[t_i]}$ is a shortest path, then $\pi_i$ is a shortest path in $G'$, hence it also is a shortest path in $G$.

\item\textbf{\ref{item:pi_i_non-croossing}} Let us assume by contradiction that there exist $i,j\in[k]$ such that $\pi_i$ and $\pi_j$ are crossing, with $i<j$. Thus $\pi_j$ has not turned always left in $X_j$, absurdum.\qed
\end{itemize}
%\vspace{-18,5pt}
\end{proof}

\subsection{Algorithm \MINNIE}\label{sec:computational_complexity}

%As we have anticipated in the previous section,
The graph $X_k$ given by the algorithm \TOPOLINO contains a shortest path for each terminal pair, but $X_k$ may also contain edges that do not belong to any shortest path. To overcome this problem we apply algorithm \MINNIE, that builds a directed graph $Y_k=\bigcup_{i\in[k]}\rho_i$, where $\rho_i$ is a directed shortest $i$-path, for $i\in[k]$. Moreover, we prove that $Y_k$ can be built in linear time. This implies that, by using the results in~\cite{err_giappo}, we can compute the length of all shortest $i$-paths, for $i\in[k]$, in $O(n)$ worst-case time (see Theorem~\ref{th:main}). 

We use the sequence of subgraphs $\{X_i\}_{i\in[k]}$. By Theorem~\ref{th:TOPOLINO}, we know that $X_i$ contains a shortest undirected $i$-path $\pi_i$ and we can list its edges in $O(|\pi_i|)$ time. But if an edge $e$ is shared by many $\pi_i$'s, then $e$ is visited many times. Thus obtaining $\bigcup_{i\in[k]}\pi_i$ by this easy procedure requires $O(kn)$ worst-case time. To overcome this problem, we should visit every edge in $\bigcup_{i\in[k]} \pi_i$ only a constant number of times.

Now we introduce two useful lemmata the will be used later. The first lemma shows that two uncomparable directed paths $\pi_i$ and $\pi_j$ (i.e., such that $i \not\prec j$ and $j \not\prec i$) in the genealogy tree $T_G$ cannot share a dart, although it is possible that $\dd{ab}\in\pi_i$ and $\dd{ba}\in\pi_j$. The second lemma deals with the intersection of non-crossing paths joining comparable pairs.

\begin{lemma}\label{lemma:preparatore_fratelli}(\cite{err_giappo})
Let $\pi_i$ be a shortest directed $i$-path and let $\pi_j$ be a shortest directed $j$-path, for some $i,j\in[k]$. If $j$ is not an ancestor neither a descendant of $i$ in $T_G$, then $\pi_i$ and $\pi_j$ have no common darts.
\end{lemma}
\begin{proof}
%Since $\{(s_i,t_i)\}_{i\in[k]}$ is well-formed and $j$ is not an ancestor neither a descendant of $i$ in $T_G$, then starting in $s_i$ and walking clockwise on $f^\infty$ we met in order $s_i$, $t_i$, $s_j$ and $t_j$. 
Let us assume by contradiction that $\pi_i$ and $\pi_j$ have some common darts, and let $d$ be the dart in $\pi_i\cap\pi_j$ that appears first in $\pi_i$. Let $R$ be the region bounded by $\uuu{\pi_j[s_j,\tail(d)]}$, $\uuu{\pi_i[s_i,\tail(d)]}$ and the clockwise undirected $s_i-s_j$ path in $f^\infty$ %It holds that $t_j\not\in V(R)$ and $\head(d)\in V(R)\setminus V(\partial R)$, hence $\pi_j$ has to go out from $R$ 
(Figure~\ref{fig:directed_non-crossing}(\subref{fig:directed_non-crossing_1}) shows $\pi_i$, $\pi_j$ and $R$).
%Path $\pi_j$ is a shortest path, so it is a simple path. Thus, in order to go out from $R$, $\pi_j$ crosses $\pi_i$ in at least one vertex in $\pi_i[s_i,\tail(d)]$. 
Being $\pi_j$ a simple path, then $\pi_j$ crosses $\pi_i$ in at least one vertex in $\pi_i[s_i,\tail(d)]$. Let $x$ be the first vertex in $\pi_i[s_i,\tail(d)]$ after $\head(d)$ in $\pi_j$. Now by looking to the cycle $\pi_i[x,\head(d)]\circ\pi_j[\head(d),x]$, it follows that $\pi_i$ and $\pi_j$ can be both shortest paths, absurdum (Figure~\ref{fig:directed_non-crossing}(\subref{fig:directed_non-crossing_2}) shows this cycle).\qed
\end{proof}

\begin{lemma}\label{lemma:preparatore_2}(\cite{err_giappo})
Let $\{\pi_i\}_{i\in[k]}$ be a set of non-crossing directed paths. Let $i,j\in[k]$, if $i$ is a descendant of $j$, then $\pi_i\cap\pi_j\subseteq\pi_\ell$, for all $\ell\in[k]$ such that $i\prec \ell\prec j$.
\end{lemma}
\begin{proof}
Let us assume $\pi_i\cap\pi_j\neq\emptyset$ and choose  $\ell\in[k]$ such that $i \prec\ell \prec j$. %Since terminal pairs are well-formed, then starting in $s_j$ and going clockwise on $f^\infty$ we met in order $s_j$, $s_\ell$, $s_i$, $t_i$, $t_\ell$ and $t_j$.
Let $e$ be the dart in $\pi_i\cap\pi_j$ that appears first in $\pi_i$ and let $Q$ be the region bounded by $\uuu{\pi_j[s_j,\tail(e)]}$, $\uuu{\pi_i[s_i,\tail(e)]}$ and the clockwise undirected $s_j-s_i$ path in $f^\infty$ %It holds that $t_\ell\not\in V(S)\setminus V(\partial R)$, hence $\pi_\ell$ has to go out from $R$ because $s_\ell\in V(S)$
(a region $Q$ and dart $e$ are shown in Figure~\ref{fig:directed_non-crossing}(\subref{fig:directed_non-crossing_3})).
It is clear that if $e\not\in\pi_\ell$, then $\{\pi_i,\pi_j,\pi_\ell\}$ is not a set of non-crossing paths, absurdum.\qed
\end{proof}

\begin{figure}[h]
\captionsetup[subfigure]{justification=centering}
\centering
%FIGURA 1
	\begin{subfigure}{4cm}
\begin{overpic}[width=4cm,percent]{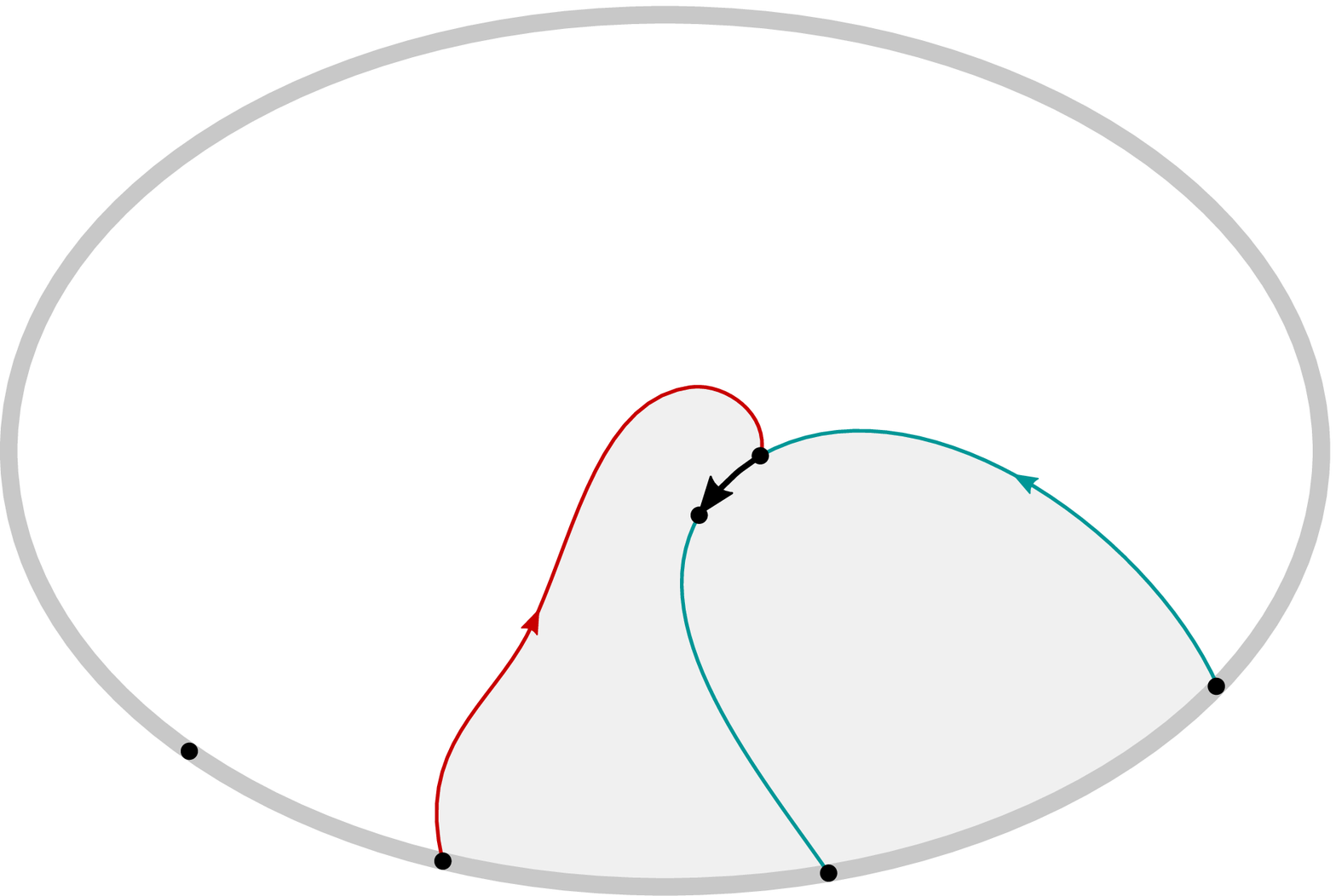}
\put(55,26){$d$}
\put(43,12){$R$}
\put(91,10){$s_i$}
\put(61,-5.5){$t_i$}
\put(31,-4){$s_j$}
\put(12,3.5){$t_j$}
\end{overpic}
%\vspace{1pt}
\caption{}\label{fig:directed_non-crossing_1}
\end{subfigure}
\qquad
%FIGURA 2
	\begin{subfigure}{4cm}
\begin{overpic}[width=4cm,percent]{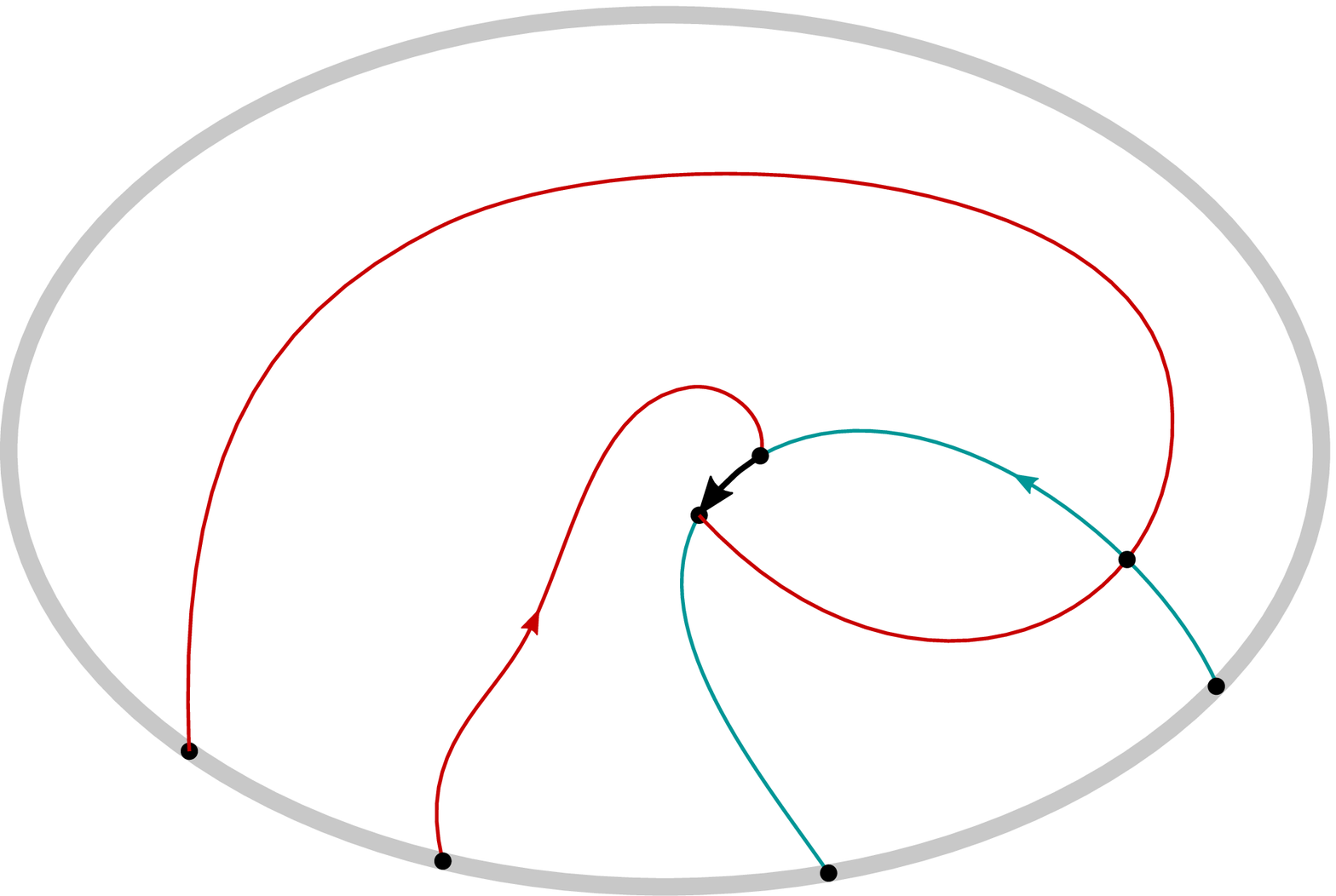}
\put(87,24){$x$}
%\put(59.5,30){$d$}
%\put(43,17){$R$}
\put(91,10){$s_i$}
\put(61,-5.5){$t_i$}
\put(31,-4){$s_j$}
\put(12,3.5){$t_j$}
\end{overpic}
\caption{}\label{fig:directed_non-crossing_2}
\end{subfigure}
\qquad
%FIGURA 3
\begin{subfigure}{4cm}
\begin{overpic}[width=4cm,percent]{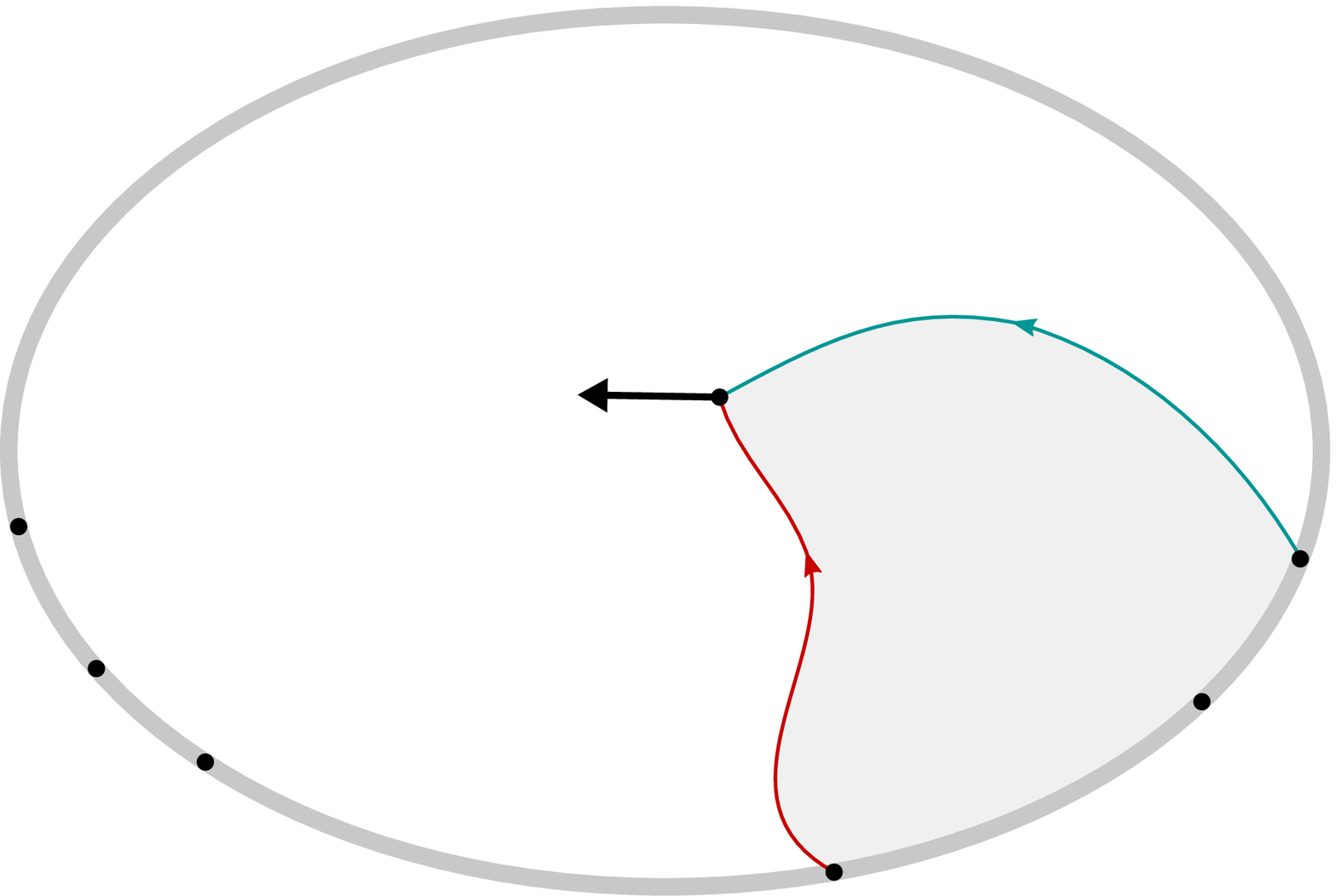}
\put(73,24){$Q$}
\put(48,39.5){$e$}
\put(61.5,-4){$s_i$}
\put(13.5,2){$t_i$}
\put(90,9.5){$s_\ell$}
\put(5,9){$t_\ell$}  
\put(97.9,21){$s_j$}
\put(-4.5,22){$t_j$}
\end{overpic}
\caption{}\label{fig:directed_non-crossing_3}
\end{subfigure}
 \caption{in (\subref{fig:directed_non-crossing_1}) and (\subref{fig:directed_non-crossing_2}) the paths $\pi_j$ and $\pi_i$, the dart $d$, the region $R$ and the vertex $x$ used in the proof of Lemma~\ref{lemma:preparatore_fratelli}. In (\subref{fig:directed_non-crossing_3}) the region $Q$ and the dart $d$ in the proof of Lemma~\ref{lemma:preparatore_2}.}
\label{fig:directed_non-crossing}
\end{figure}

%\begin{figure}[h]
%\centering
%\begin{overpic}[width=4cm,percent]{images/directed_non-crossing_3.eps}
%\put(73,24){$S$}
%\put(47,39.5){$e$}
%\put(61.5,-4){$s_i$}
%\put(13.5,2){$t_i$}
%\put(90,9.5){$s_\ell$}
%\put(5,9){$t_\ell$}  
%\put(97.9,21){$s_j$}
%\put(-4.5,22){$t_j$}
%\end{overpic}
%\caption{the region $S$ and the dart $e$ in the proof of Lemma~\ref{lemma:preparatore_2}.
%}
%\label{fig:directed_non-crossing_3}
%\end{figure}

Now we show how to use these two lemmata for our goals. Let $\rho_i$ be a shortest directed $i$-path and let $\rho_j$ be a shortest directed $j$-path, for some $i,j\in[k]$, $i\neq j$. By Lemma~\ref{lemma:preparatore_fratelli}, if $i$ and $j$ are not comparable in $T_G$, then $\rho_i$ and $\rho_j$ have no common darts. Moreover, by Lemma~\ref{lemma:preparatore_2}, if $i$ is an ancestor of $j$ in $T_G$, then $\rho_i\cap\rho_j\subseteq\rho_{p(j)}$. By using these two facts, in order to list darts in $\rho_{i}$,  then it suffices to find darts in $\rho_i\setminus\rho_{p(i)}$,
for all $i \in [k] \setminus \{1\}$ (we remind that 1 is the root of $T_G$). To this goal we use algorithm \MINNIE, that builds a sequence of directed graphs $\{Y_i\}_{i\in[k]}$ such that $Y_k$ is equal to $\bigcup_{i\in[k]}\rho_i$, where $\rho_i$ is a shortest directed $i$-path, for $i\in[k]$.

We prove the correctness of algorithm \MINNIE in Theorem~\ref{th:MINNIE}. At iteration $i$ we compute $\rho_i\setminus\rho_{p(i)}$, showing that $\rho_i\setminus\rho_{p(i)}=\sigma_i\cup\rev[\tau_i]$, where $\sigma_i$ and $\tau_i$ are computed in Line~\ref{line:sigma_i} and Line~\ref{line:tau_i}, respectively. We observe that if $\rho_i\cap\rho_{p(i)}=\emptyset$, then $\sigma_i=\rev[\tau_i]=\rho_i$.

\begin{figure}
\begin{algorithm}[H] 
\SetAlgorithmName{Algorithm \texttt{NCSPunion}}{}{}
\renewcommand{\thealgocf}{}
 \caption{}
 \KwIn{an undirected unweighted planar embedded graph $G$ and $k$ well-formed terminal pairs of vertices $(s_i,t_i)$, for $i\in[k]$, on the external face of $G$}
 \KwOut{a directed graph $Y_k$ formed by the union of directed shortest non-crossing paths from $s_i$ to $t_i$, for $i\in[k]$}
{Compute $X_1$ as in algorithm \TOPOLINO\;
$Y_1$ is the directed version of $X_1$ oriented from $s_1$ to $t_1$\label{line:1_compute_pi_1}\;
\For{$i=2,\ldots,k$}{
Compute $X_i$ as in algorithm \TOPOLINO\;
$\sigma_i$ is the directed path that starts in $s_i$ and always turns left in $X_i$ until either $\sigma_i$ reaches $t_i$ or the next dart $d_i$ of $\sigma_i$ satisfies $d_i\in Y_{i-1}$\label{line:sigma_i}\;
$\tau_i$ is the directed path that starts in $t_i$ and always turns right in $X_i$ until either $\tau_i$ reaches $s_i$ or the next dart $d_i'$ of $\tau_i$ satisfies $\rev[d_i']\in Y_{i-1}$\label{line:tau_i}\;
$Y_i=Y_{i-1}\cup\sigma_i\cup \rev[\tau_i]$\label{line:Y}\;
}
}
\end{algorithm}
\end{figure}

\begin{theorem}\label{th:main}
Algorithm \MINNIE has $O(n)$ worst-case time complexity.
\end{theorem}
\begin{proof}
Algorithm \MINNIE uses algorithm \TOPOLINO, that requires $O(n)$ worst-case time by Lemma~\ref{lemma:TOPOLINO_O(n)}. Moreover, algorithm \MINNIE visits each dart of  the ``directed version" of $X_k$ at most $O(1)$ times, where the \emph{directed version of} $X_k$ is the directed graph built from $X_k$ by replacing each edge $ab$ by the pair of darts $\dd{ab}$ and $\dd{ba}$. Thus, algorithm \MINNIE requires $O(n)$ worst-case time, since $X_k$ is a subgraph of $G$.\qed
\end{proof}

\begin{theorem}\label{th:MINNIE}
Graph $Y_k$ computed by algorithm \MINNIE is the union of $k$ shortest non-crossing $i$-paths, for $i\in[k]$.
\end{theorem}
\begin{proof}
Let $\{\pi_i\}_{i\in[k]}$ be the set of paths defined in Theorem~\ref{th:TOPOLINO}. For all $i\in[k]$, we denote by $\dd{\pi_i}$ the directed version of $\pi_i$, oriented from $s_i$ to $t_i$.

First we define $\rho_1=\dd{\pi_1}$ and for all $i\in[k]\setminus\{1\}$ we define
\begin{equation}
\label{eq:rho_i}
\rho_i=
\begin{cases}
\dd{\pi_i}[s_i,u_i]\circ \rho_{p(i)}[u_i,v_i]\circ \dd{\pi_i}[v_i,t_i], &\text{ if }\,\, \dd{\pi_i}\cap\rho_{p(i)}\neq\emptyset,\\
\dd{\pi_i}, &\text{ otherwise},
\end{cases}
\end{equation}
where we assume that if $V(\dd{\pi_i}\cap\rho_{p(i)})\neq\emptyset$, then $u_i$ and $v_i$ are the vertices in $V(\dd{\pi_i}\cap\rho_{p(i)})$ that appear first and last  in $\dd{\pi_i}$, respectively; the definition of $\rho_i$ as in \eqref{eq:rho_i} is shown in Figure~\ref{fig:eq_rho_i}. Now we split the proof into three parts: first we prove that $\{\rho_i\}_{i\in[k]}$ is a set of shortest paths (we need it to apply Lemma~\ref{lemma:preparatore_fratelli}); second we prove that $\{\rho_i\}_{i\in[k]}$ is a set of non-crossing paths (we need it to apply Lemma~\ref{lemma:preparatore_2}); third we prove that $Y=\bigcup_{i\in[k]}\rho_i$ (we prove it by Lemma~\ref{lemma:preparatore_fratelli} and Lemma~\ref{lemma:preparatore_2}).

\begin{itemize}
\item \textbf{$\{\rho_i\}_{i\in[k]}$ is a set of shortest paths:} we proceed by induction on $i$. The base case is trivial because $\pi_1$ is a shortest path by definition. Let us assume that $\rho_j$ is a shortest $j$-path, for $j<i$, we have to prove that $\rho_i$ is a shortest $i$-path. If $\dd{\pi_i}\cap\rho_{p(i)}=\emptyset$, then 
$\rho_i=\dd{\pi_i}$ by \eqref{eq:rho_i}, thus the thesis holds because $\{\pi_i\}_{i\in[k]}$ a set of shortest paths. Hence let us assume that $\dd{\pi_i}\cap\rho_{p(i)}\neq\emptyset$, then it suffices, by definition of $\rho_i$, that $|\pi_i[u_i,v_i]|=|\rho_{p(i)}[u_i,v_i]|$. It is true by induction.

\item \textbf{$\{\rho_i\}_{i\in[k]}$ is a set of non-crossing paths:} we proceed by induction on $i$. The base case is trivial because there is only one path. Let us assume that $\{\rho_j\}_{j\in[i-1]}$ is a set of non-crossing paths, we have to prove that $\rho_i$ does not cross $\rho_j$, for any $j<i$. 

If $\rho_i$ and $\rho_j$ are crossing and $j$ is not an ancestor of $i$, then, by construction of $\rho_i$, either $\rho_{p(i)}$ and $\rho_j$ are crossing or $\pi_i$ and $\pi_j$ are crossing; that is absurdum in both cases by induction and Theorem~\ref{th:TOPOLINO}. Moreover, by definition, $\rho_i$ does not cross $\rho_{p(i)}$, and by induction, if $\ell$ is an ancestor of $i$ such that $\ell\neq p(i)$, then $\rho_i$ does not cross $\rho_\ell$, indeed, if not, then $\rho_\ell$ would cross $\rho_{p(i)}$, absurdum. Hence $\{\rho_i\}_{i\in[k]}$ is a set of non-crossing paths.

\item \textbf{$Y$ is the union of $\rho_i$'s:} now we prove that $Y=\bigcup_{i\in[k]}\rho_i$. In particular we show that $\rho_1=\dd{\pi_1}$ and for all $i\in[k]\setminus\{1\}$ 

\begin{equation}
\label{eq:rho_i_2}
\rho_i=
\begin{cases}
\sigma_i\circ \rho_{p(i)}[u_i,v_i]\circ\rev[\tau_i], &\text{ if }\,\, \dd{\pi_{i}}\cap\rho_{p(i)}\neq\emptyset,\\
\dd{\pi_i}, &\text{ otherwise}.
\end{cases}
\end{equation}

\noindent
Again, we proceed by induction on $i$. The base case is trivial, thus we assume that (\ref{eq:rho_i}) is equivalent to (\ref{eq:rho_i_2}) for all $i<\ell$. We have to prove that (\ref{eq:rho_i}) is equivalent to (\ref{eq:rho_i_2}) for $i=\ell$.

If $\dd{\pi_\ell}$ does not intersect any dart of $\rho_{p(\ell)}$, then \eqref{eq:rho_i} is equivalent to \eqref{eq:rho_i_2}. Thus we assume that $\dd{\pi_\ell}\cap\rho_{p(\ell)}\neq\emptyset$. By (\ref{eq:rho_i}) and (\ref{eq:rho_i_2}) and by definition of $\sigma_i$ and $\tau_i$ in Line~\ref{line:sigma_i} and Line~\ref{line:tau_i}, respectively, it suffices to prove that $d_i\in \rho_{p(i)}$ and $\rev[d_i']\in\rho_{p(i)}$. 

Now, by induction we know that $d_i\in\rho_\ell$ for some $\ell<i$, we have to show that $d_i\in \rho_{p(i)}$. By Lemma~\ref{lemma:preparatore_fratelli} and being $\{\rho_j\}_{j\in[k]}$ a set of shortest paths, it holds that $\ell$ is an ancestor or a descendant of $i$. Being the $s_j$'s visited clockwise by starting from $s_1$, then $\ell$ is an ancestor of $i$. Finally, by Lemma~\ref{lemma:preparatore_2} and being $\{\rho_j\}_{j\in[k]}$ a set of non-crossing path, it holds that $\rho_i\cap\rho_\ell\subseteq\rho_{p(i)}$. Being $p(i)<i$, then $d_i\in\rho_{p(i)}$ as we claimed. By a similar argument, it holds that $\rev[d_i']\in\rho_{p(i)}$.\qed
\end{itemize}
%\vspace{-22pt}
\end{proof}

\begin{figure}[h]
\captionsetup[subfigure]{justification=centering}
\centering
%FIGURA 1
	\begin{subfigure}{4.5cm}
\begin{overpic}[width=4.5cm,percent]{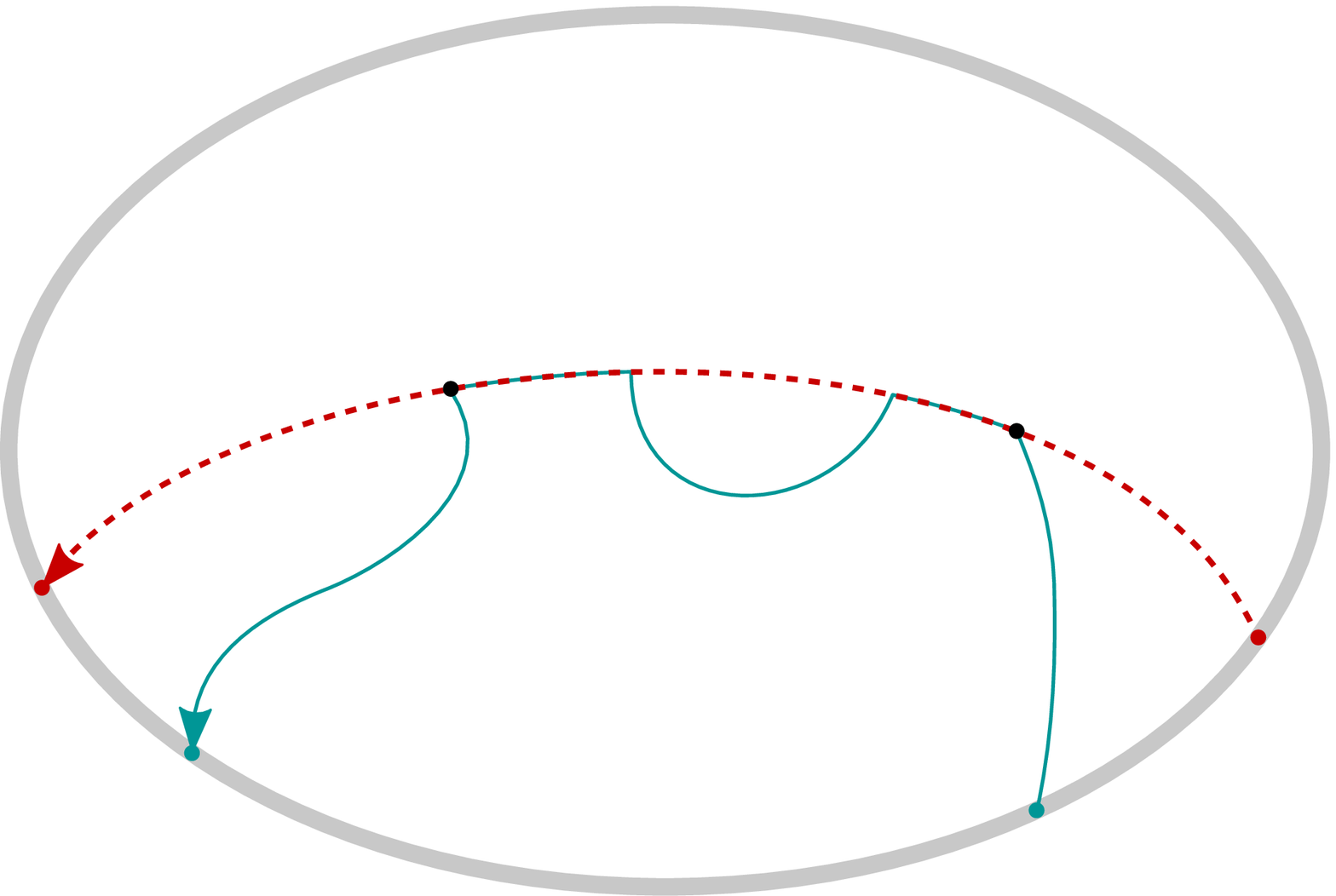}
\put(76,1){$s_i$}
\put(13,4){$t_i$}
\put(95,15.2){$s_{p(i)}$}
\put(-8,17.5){$t_{p(i)}$}
\put(51,21.5){$\dd{\pi_i}$}
%\put(5.5,37.5){$\dd{\pi}_{\!\!\!p(i)}$}
\put(51,43){$\dd{\pi}_{\!\!p(i)}$}

\put(74.5,37){$u_i$}
\put(30.5,40.5){$v_i$}
\end{overpic}
\caption{path $\dd{\pi_i}$ and $\dd{\pi}_{\!\!p(i)}$}\label{fig:eq_rho_i_a}
\end{subfigure}
\qquad\quad
%FIGURA 4
	\begin{subfigure}{4.5cm}
\begin{overpic}[width=4.5cm,percent]{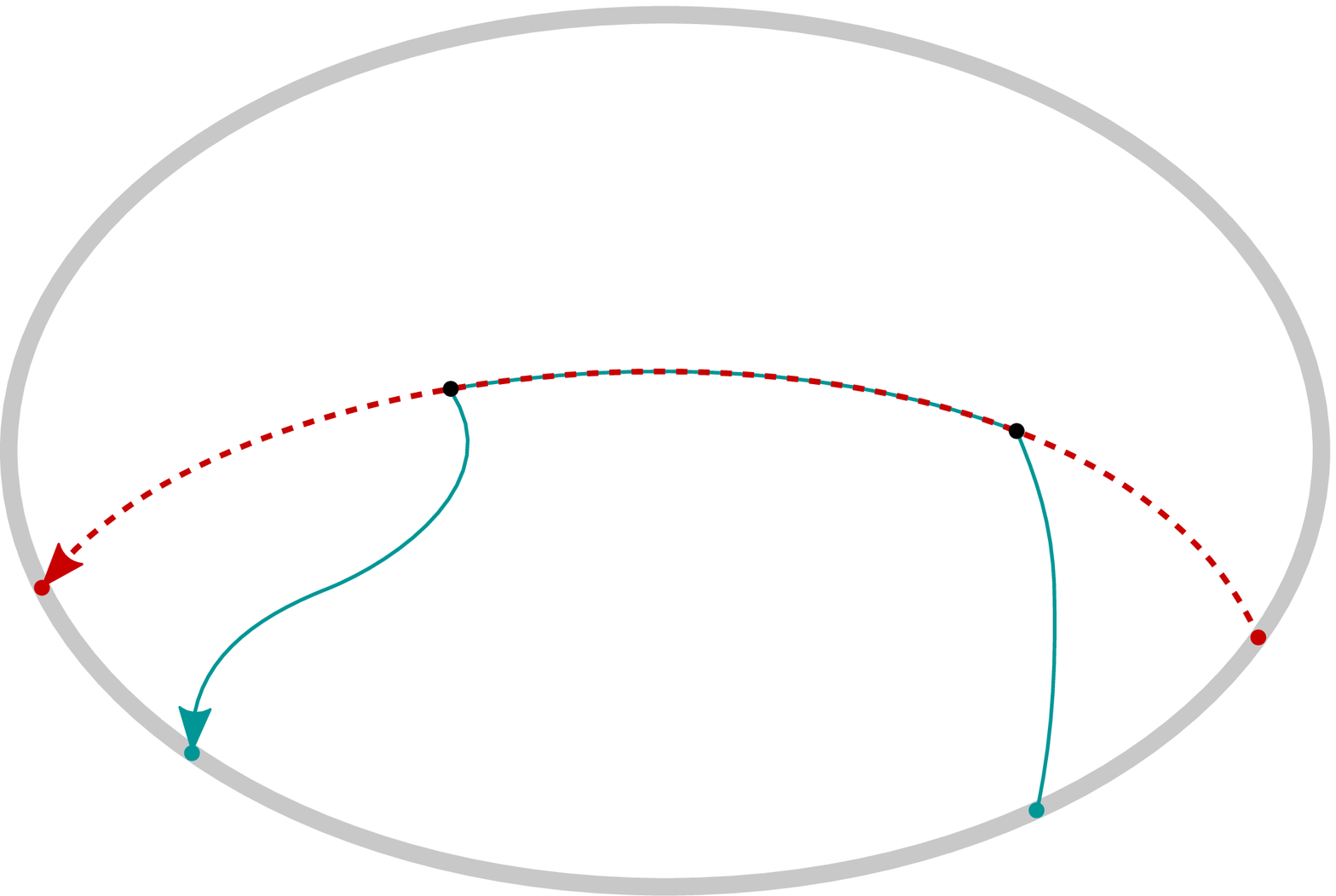}
\put(76,1){$s_i$}
\put(13,4){$t_i$}
\put(95,15.2){$s_{p(i)}$}
\put(-8,17.5){$t_{p(i)}$}

\put(29,22){$\rho_i$}

%\put(74.5,37){$u_i$}
%\put(30.5,40.5){$v_i$}
\end{overpic}
\caption{path $\rho_i$}\label{fig:eq_rho_i_b}
\end{subfigure}	
  \caption{proof of Theorem~\ref{th:MINNIE}, explanation of \eqref{eq:rho_i}.}
\label{fig:eq_rho_i}
\end{figure}

It is proved in~\cite{err_giappo} that, starting from the union of a set of shortest (not necessarily non-crossing) paths between well-formed terminal pairs,
distances between terminal pairs can be computed in linear time. Thus we can give the following main theorem.

\begin{theorem}\label{th:main}
Given an undirected unweighted plane graph $G$ and a set of well-formed terminal pairs $\{(s_i,t_i)\}$ on the external face $f^\infty$ of $G$ we can compute $U=\bigcup_{i\in[k]}p_i$ and the lengths of all $p_i$, for $i\in[k]$, where $p_i$ is a shortest $i$-path and $\{p_i\}_{i\in[k]}$ is a set of non-crossing %single-touch 
paths, in $O(n)$ worst-case time.
\end{theorem}
\begin{proof}
By Theorem~\ref{th:MINNIE}, the required graph $U$ is the undirected version $\uuu{Y_k}$ of the graph computed by algorithm \MINNIE, that has $O(n)$ worst-case time complexity. Moreover, we compute the length of $p_i$, for all $i\in[k]$, in $O(n)$ worst-case time by using the results in~\cite{err_giappo}.\qed
\end{proof}

\begin{remark}
For graphs with small integer weights, we can obtain all the previous results in $O(n+L)$ worst-case time, where $L$ is the sum of all edge weights, by splitting an edge of weight $r$ in $r$ unweighted edges.
\end{remark}

\section{Conclusions}\label{sec:conclusions}
In this paper we have shown a linear time algorithm to compute the union of non-crossing shortest paths whose extremal vertices are in the external face of an undirected unweighted  planar graph.

The algorithm relies on the algorithm by Eisenstat and Klein for computing SSSP trees rooted on the vertices of the external face and on the novel concept of ISP subgraph of a planar graph, that can be of interest itself. 
The same approach cannot be extended to  weighted graphs,  because the algorithm of Eisenstat and Klein works only in the unweighted case.

As stated in \cite{erickson-nayyeri} our results may be applied in the case of terminal pairs lying on $h$ face boundaries.

We wish to investigate the non-crossing shortest paths problem when each terminal pair contains only one vertex on the external face.

\bibliographystyle{siam}
\bibliography{biblio.bib}

\end{document}